\documentclass[11pt,times]{article}
\usepackage[utf8]{inputenc}
\usepackage{amssymb,amsmath}
\usepackage{algorithm}
\usepackage{algpseudocode}
\usepackage{amsthm}
\usepackage{enumerate}
\usepackage{bold-extra}

\usepackage{tikz}
\usetikzlibrary{decorations.pathmorphing}
\usepackage[center,scriptsize]{subfigure}
\usetikzlibrary{shapes,fit}

\usepackage{bm}

\usepackage{comment}
\excludecomment{ignore}

\newtheorem{theorem}{Theorem}[section]

\newtheorem{definition}[theorem]{Definition}

\newtheorem{lemma}[theorem]{Lemma}

\newcommand{\Xomit}[1]{}
\newcommand{\vone}{\vspace{.1in}}
\newcommand{\vh}{\vspace{.05in}}

\newcommand{\vmone}{\vspace{-.1in}}

\newcommand{\MarkEdges}{Extensions}
\newcommand{\MinWeightDelta}{Min-Wt-Cyc}

\algtext*{EndWhile}
\algtext*{EndFor}
\algtext*{EndIf}

\advance \topmargin by -\headheight
\advance \topmargin by -\headsep
\evensidemargin \oddsidemargin
\begin{document}
\title{Finding $k$ Simple Shortest Paths and Cycles}
\author{Udit Agarwal $^*$ and Vijaya Ramachandran~\thanks{Dept. of Computer Science, University of Texas, Austin TX 78712. Email: {\tt udit@cs.utexas.edu, vlr@cs.utexas.edu}. This work was supported in part by NSF Grant CCF-1320675. The first author's research was also supported in part by a Calhoun Fellowship.}}
\maketitle

\vspace{-.5in}
\begin{abstract}
We present algorithms and hardness results for several problems related to finding multiple simple shortest paths in a graph.
Our main result is a new algorithm for finding $k$ simple shortest paths for all pairs of vertices
in a weighted directed graph $G=(V,E)$. For $k=2$ our algorithm runs in $O(mn + n^2 \log n)$ time where $m$ and $n$ are the number of edges
and vertices in $G$.  
Our approach is based on forming suitable path extensions to find simple shortest paths; this method is different from the `detour finding' technique used in most of the prior work on simple shortest paths, replacement paths, and distance sensitivity oracles.
We complement this result by showing that finding 2 simple shortest paths even for a single pair of vertices is at least as hard as
finding
 a minimum weight cycle in $G$, for which no sub-$mn$ time algorithm is known.

 We present new algorithms for generating  simple cycles and simple paths in $G$  in non-decreasing order of their weight. The algorithm for
generating simple paths is much faster,  and uses another variant of path extensions.  We also give  hardness results for sparse graphs, relative to the complexity of computing a minimum weight cycle in a graph, for several variants of problems related to finding $k$ simple paths and cycles, and we give related results for undirected graphs.
\end{abstract}

\section{Introduction}\label{sec:intro}

Computing shortest paths in a weighted directed graph 
 is a very well-studied problem. Let $G=(V,E)$ be a directed graph with
non-negative edge weights, with $|V|=n$, $|E|=m$.
Then, a shortest path for a single pair of vertices can be computed in $\tilde{O}(m)$
 time, and
the all pairs shortest paths (APSP) in $\tilde{O}(mn)$ time~\cite{CLRS09}, where $\tilde{O}$ hides $polylog(n)$ factors.

A related problem is one
of computing a sequence of $k$ shortest paths, for $k> 1$.
If the paths need not be simple, the problem of generating
$k$ shortest paths is also well understood, and the most efficient algorithm is due to
Eppstein~\cite{Eppstein98}, which has the following bounds --- $O(m+n\log n+k)$ for a single pair of vertices
and $O(m + n\log n + kn)$ for single source.

It is noted in~\cite{Eppstein98}
that the simple paths version of the $k$ shortest paths problem is more common than the unconstrained version considered in~\cite{Eppstein98}.
In the $k$ simple shortest paths problem, given a pair of vertices $s,t$,  the output is a sequence of $k$
simple paths from $s$ to $t$, where the $i$-th path in the collection is a shortest simple path in the graph 
that is not identical to any of the $i-1$ paths preceding it in the output.
(Note that these $k$ simple shortest paths need not have the same weight).

The all-pairs version of this problem (where the paths need not be simple) was considered in the classical papers of 
 Lawler~\cite{Lawler72,Lawler77} and Minieka~\cite{Minieka74},
 and the most efficient current algorithm for $k$-APSP runs the SSSP algorithm in~\cite{Eppstein98} on each of the
$n$ vertices in turn. In this  paper, a central problem we consider the all-pairs version of this problem when the paths are required to 
be simple ($k$-APSiSP). It was noted in Minieka~\cite{Minieka74} that the all-pairs version of $k$ shortest paths
 becomes significantly harder when simple paths are required.

 Even for a single source-sink pair, the problem of generating $k$ simple shortest paths   ($k$-SiSP) is considerably more challenging than the unrestricted version considered in~\cite{Eppstein98}.
 Yen's algorithm~\cite{Yen71} finds the
 $k$ simple shortest paths for a specific pair of vertices  in
 $O(k \cdot (mn + n^2 \log n))$. 
  Gotthilf and Lewenstein \cite{GL09} improved the time bound slightly to $O(k(mn + n^2\log \log n))$.
  In terms of hardness of this problem,
   it is shown by- V. Williams and R. Williams ~\cite{WW10}
 that if the {\it second} simple shortest path for a single source-sink pair (i.e., $k=2$ in $k$-SiSP)
 can be found in $O(n^{3 - \delta})$ time for some $\delta>0$,
 then APSP can also be computed in $O(n^{3-\alpha})$ time for some $\alpha >0$; the latter is a major open problem.
(In this formulation, the dependence on $m$, the number of edges, is not considered.)

The $k$-SiSP problem is much simpler in the undirected case and is known to be solvable in $O(k(m + n\log n))$ time \cite{KIM82}. For unweighted directed graphs, Roditty and Zwick \cite{RZ12} gave an $\tilde{O}(km\sqrt{n})$ randomized algorithm for $k$-SiSP. They also showed that $k$-SiSP can be solved with $O(k)$ executions of an algorithm for the $2$-SiSP problem.
Approximation algorithms for $k$-SiSP are considered in~\cite{Roditty10,Bernstein10}. Other
related results can be found in~\cite{Feng14,FR15,HC99,HMS07,HS01,HSB07,LP97,MK05,MP03,Perko86,Sedeno12}.

A problem related to 2-SiSP is the {\it replacement paths} problem. In the $s$-$t$ version of this problem, we need to output
a shortest path from $s$ to $t$ when an edge on the shortest path $p$  is removed; the output is a collection of  $|p|$ paths, each a
shortest path from $s$ to $t$ when  an edge on $p$ is removed. Clearly, given the solution to the $s$-$t$ replacement paths
problem, the second shortest path from $s$ to $t$ can be computed as the path of minimum weight in this solution. This
is essentially the method used in all algorithms for 2-SiSP (and with modifications, for $k$-SiSP), and thus the current fastest algorithms for  2-SiSP and replacement paths have the same time bound.
 For the all-pairs case that is of interest to us, the output for the replacement paths
problem would be $O(n^3)$ paths, where each path is shortest for a specific vertex pair, when a specific edge in its
shortest path is removed. In view of the large space needed for this output, in the all-pairs version of  replacement paths,
the problem of interest is {\it distance sensitivity oracles (DSO).} Here, the output is a compact representation from which any specific
replacement path can be found with $O(1)$ 
time.
The first such  oracle  was 
developed in Demetrescu et. al. \cite{DTCR08}, and it has size $O(n^2 \log n)$. 
The current best construction time for an oracle of this size is 
$O(mn \log n + n^2 \log^2 n)$ time for a randomized algorithm, 
and a log factor slower for a deterministic
algorithm,
by Bernstein and Karger~\cite{BK09}.
Given such an oracle, the output to
$2$-APSiSP can be computed with $O(n)$ queries for each
source-sink pair, i.e.,
with  $O(n^3)$ queries to the DSO.

\vone
The most well-studied problem relating to $k$ simple  cycles is that of finding cycles in the overall graph.
The problem of enumerating simple (or {\it elementary}) cycles --- in no particular order ---has been
 studied extensively~\cite{Tie70, Wei72,Tar73,Joh75}, and the first algorithm that generated successive  simple cycles in polynomial time was given by Tarjan~\cite{Tar73}; this algorithm generates each successive cycle in $O(mn)$ time. 
 This result was improved to linear time by Johnson~\cite{Joh75}. We are not aware of prior bounds on enumerating simple cycles in nondecreasing order of weight. 
 
 Another well-studied problem relating to simple cycles is the problem of finding
a minimum weight cycle  (Min-Wt-Cyc) in a graph.
 Finding a minimum weight cycle  in either a directed or an undirected graph
is known to be equivalent to APSP for sub-cubic algorithms ~\cite{WW10}. The time bound for 
Min-Wt-Cyc in sparse graphs is $\tilde{O}(m \cdot n)$, 
and finding a faster algorithm for it is a long-standing open question.

In this paper, we concentrate on results for truly sparse graphs with arbitrary non-negative edge weights. Hence we do not consider results for small integers weights or for dense graphs; several
subcubic results for such graphs are known using fast matrix multiplication.

\subsection{Our Contributions}

We present several algorithmic results, and complement many of them with hardness results relative to computing Min-Wt-Cyc  on sparse graphs.
\footnote{Except for $k$-All-SiSP (see Section~\ref{sec:allsisp2}), we can also handle negative edge-weights as long as there are no
negative-weight cycles, by applying Johnson's transformation~\cite{Joh77} to obtain an equivalent input with nonnegative edge weights.
If the resulting edge-weights include weight 0, we will
use the pair $(wt(p), len(p))$ as the weight for  path $p$, where $len(p)$ is the number of edges
in it; this causes the weight of a proper subpath of $p$ to be smaller  than the weight of $p$.}

\vone
1. {\it Computing $k$ simple shortest paths for all pairs ($k$-APSiSP) in $G$.}
We present a new approach to the $k$-APSiSP problem.
In order to construct
the desired set $P^*_k (x,y)$ of $k$ simple shortest paths from $x$ to $y$,
our method uses the notion of a `nearly $k$ SiSP set' $Q_k(x,y)$,
defined as follows.

\begin{definition}\label{def:qk}
Let $G=(V,E)$ be a directed graph with nonnegative edge weights. For $k \geq 2$,
and a vertex pair $x,y$, let $k^* = \min\{r,k\}$, where $r$ is the number of simple paths from $x$ to $y$ in $G$.
Then,\\
(i)  $P^*_k (x,y)$ is the set of $k^*$ \emph{simple shortest paths from $x$ to $y$} in $G$.\\
(ii) $Q_k(x,y)$ is the set of \emph{$k$ nearly simple shortest paths from $x$ to $y$}, 	defined as follows.
If $k^*=k$  and the $k-1$ simple shortest paths from $x$ to $y$ share the same first edge $(x,a)$ then $Q_k(x,y)$
contains these $k-1$ simple shortest paths, together with the 
simple shortest path from $x$ to $y$ that does not start with edge
$(x,a)$, if such a path exists. Otherwise
 (i.e, if either the former or latter condition does not hold),
, $Q_k(x,y) = P^*_k (x,y)$.
\end{definition}

Our algorithm for $k$-APSiSP  first constructs  $Q_k(x,y)$
for all pairs of vertices $x,y$, and then uses these sets in an efficient
algorithm, {\sc Compute-APSiSP}, to compute the $P^*_k (x,y)$  for all $x,y$. The latter algorithm runs in time
$O(k \cdot n^2 + n^2 \log n)$ for any $k$, while our method for constructing the $Q_k(x,y)$  depends on $k$. 
For $k=2$ we present an $O(mn + n^2 \log n)$ time method to compute the $Q_2(x,y)$ sets;
this gives a 2-APSiSP algorithm
that matches Yen's bound of $O(mn + n^2 \log n)$ for
$2$-SiSP for a single pair. 
It is also faster (by a polylogarithmic factor)
than the best algorithm for DSO (distance sensitivity
oracles) for the all-pairs
replacement paths problem~\cite{BK09}.
In fact, we also show that the $Q_2(x,y)$ sets can be computed in $O(n^2)$ time using a DSO, and hence 2-APSiSP can be computed in
$O(n^2 \log n)$ time plus the time to construct the DSO. 

 For $k \geq 3$ our algorithm to compute the $Q_k$ sets makes calls to an algorithm for 
$(k-1)$-APSiSP, so we combine the two components  together in a single  recursive method, {\sc APSiSP}, that takes as
input $G$ and $k$, and outputs the $P^*_k$ sets for all vertex pairs. The time bound for {\sc APSiSP} 
increases with $k$: it is faster than Yen's method for $k=3$
by a factor of $n$ (and hence is faster than the current fastest method by almost a factor of $n$), 
it matches Yen for $k=4$, and its
performance degrades for larger $k$.

Our method for computing $k$-APSiSP (using the $Q_k(x,y)$ sets) extends an existing simple path in the
data structure to a new simple path in the data structure by adding a single incoming edge.
 These extensions have to be
performed carefully in order to ensure that the extended path is simple, and the collection of paths
formed
includes the $k$-SiSPs for every pair of vertices. This approach differs from all previous approaches to $k$ simple paths and replacement paths. All known previous algorithms for 2-SiSP compute replacement paths for every edge on the shortest path (by computing suitable `detours').
In fact, Hershberger et al.~\cite{HSB07} present a lower bound for $k$-SiSP, exclusively for the class of algorithms that use detours, by pointing out that
all known algorithms for $k$-SiSP compute replacement paths, and all known replacement path algorithms use detours. 
Although this lower bound is for $k$-SiSP and not $k$-APSiSP, the model used for this lower bound does not apply to our algorithm, since 
 {\sc Compute-APSiSP} only computes a subset of the replacement paths (across all vertex pairs), and further, it can generate and inspect paths that are not detours, including paths with cycles.  Thus our method is fundamentally new.
Our algorithms for $k$-APSiSP are presented in
Section~\ref{sec:apsisp}.

\begin{table*}
\small
\centering
\begin{tabular}{| c | l | l |}
\hline
{\sc Problem} & \hfill {\sc Known Results}  \hfill $~$ & \hfill  \textbf{\sc New} {\sc Results} \hfill $~$ \\
\hline
\hline
2-APSiSP & \underline{Upper Bound}: $O(n^3)$ (using DSO) & \textbf{\underline{Upper Bound}:} $\bm{\tilde{O}(mn)}$  \vspace{0.5mm} \\
 (Sec. \ref{sec:2APSiSP})  & & \\
\hline
3-APSiSP & \underline{Upper Bound}: $\tilde{O}(mn^3)$ ~\cite{Yen71} & \textbf{\underline{Upper Bound}:} $\bm{\tilde{O}(mn^2)}$ \vspace{0.5mm} \\
 (Sec. \ref{sec:k-apsisp})  & & \\
\hline
2-SiSP & \underline{Hardness}: Min-Wt-$\Delta$ $\leq$ 2-SiSP &  \\
 (Sec.~\ref{sec:intro}, \ref{sec:hardness}) & (for subcubic) ~\cite{WW10} & \textbf{\underline{Hardness}:} \textbf{\MinWeightDelta}  $\bm{\leq_{(m+n)}}$ \textbf{2-SiSP} \\
& \underline{Upper Bound}: $\tilde{O}(mn)$ ~\cite{GL09} & \vspace{0.5mm} \\
\hline
$k$-SiSP  & \underline{Hardness}: Min-Wt-$\Delta$ $\leq$ $k$-SiSP & \textbf{\underline{Hardness}:}  \textbf{Min-Wt-Cyc} $\bm{\leq_{(m+n)}}$ $\bm{k}$\textbf{-SiSP} \vspace{.5mm} \\
(for $k\geq2$) & (for subcubic) ~\cite{WW10} & \textbf{(improvement due to above result)} \\
 (Sec. 1, \ref{sec:hardness})
 & \underline{Upper Bound}: $\tilde{O}(kmn)$ ~\cite{GL09} &  \vspace{0.5mm} \\
\hline
$k$-SiSC & \hfill --- \hfill $~$& $\bm{k}$\textbf{-SiSP} $\bm{\equiv_{(m+n)}}$ $\bm{k}$\textbf{-SiSC} \\
 (Sec. \ref{sec:k-cycles}, \ref{sec:hardness})  & & \\
\hline
\vspace{0.5mm} 
$k$-AVSiSC & \hfill --- \hfill $~$ & \textbf{\underline{Hardness}:} \textbf{\MinWeightDelta} $\bm{\leq_{(m+n)}}$ $\bm{2}$\textbf{-AVSiSC} \\
 (Sec. \ref{sec:k-cycles}, \ref{sec:hardness}) & & \textbf{\underline{Upper Bound}:} $\bm{\tilde{O}(mn)}$ \textbf{for} $\bm{(k = 2)}$ \\
& & \hspace{15mm} \textbf{and} $\bm{\tilde{O}(kmn^2)}$ \textbf{for} $\bm{(k > 2)}$ \\ 
\hline
$k$-All-SiSC & \hfill --- \hfill $~$ & \textbf{\underline{Hardness}:} \textbf{\MinWeightDelta} $\bm{\leq_{(m+n)}}$ $\bm{2}$\textbf{-All-SiSC} \\
 (Sec. \ref{sec:allsisp1}, \ref{sec:hardness}) & & \textbf{\underline{Upper Bound}:} $\bm{\tilde{O}(mn)}$ {\bf per cycle} \vspace{0.5mm}  \\
\hline
\vspace{0.5mm} 
$k$-All-SiSP & \hfill --- \hfill $~$  & \textbf{\underline{Upper Bound}:} \textbf{amortized} $\bm{\tilde{O}(k)}$ \textbf{if} $\bm{k < n}$ \\
 (Sec. \ref{sec:allsisp2}) & & \hspace{20mm} \textbf{and} $\bm{\tilde{O}(n)}$ \textbf{if} $\bm{k \geq n}$ \textbf{per path} \\
& & \hspace{20mm} \textbf{after a startup cost of}  $\bm{O(m)}$ \\
\hline

\end{tabular}
\caption{Our Results for directed graphs.$~$(DSO stands for Distance Sensitivity Oracles; $\leq_{(m+n)}$ reductions are defined in Section~\ref{sec:hardness}, as are results for undirected graphs.)} \label{table1}
\end{table*}

\vone
2. {\it Generating the $k$ simple shortest cycles ($k$-All-SiSC) or $k$ simple shortest paths ($k$-All-SiSP) in $G$.}
In Section~\ref{sec:allsisp} we consider the problem of enumerating  simple cycles and paths
in the graph, in nondecreasing order of their weights. 
We present an algorithm for $k$-All-SiSC that, after an
initial preprocessing cost of $O(mn + n^2 \log n)$, generates each successive simple shortest cycle in $G$ in $O(APSP) = O(mn + n^2 \log\log n)$ time. We also observe (in Section~\ref{sec:hardness}) that it is unlikely that we can
obtain the linear time achieved for generating successive simple cycles in no particular order,  since we show that generating each successive simple shortest cycle is at least as hard as
Min-Wt-Cyc.

Complementing the result for $k$-All-SiSC, we present an algorithm for $k$-All-SiSP that generates each successive simple path in $\tilde{O}(k)$ time if $k<n$, and in $\tilde{O}(n)$ time if $k>n$, after
an initial start-up cost of $O(m)$ to find the first path.  
Here, $\tilde{O}$ omits $polylog(n)$ factors.
This time bound is considerably faster than that for $k$-All-SiSC and for Min-Wt-Cyc.
Our method, {\sc All-SiSP}, is again one of extending existing
paths by an edge;
it is, however, different from 
{\sc Compute-APSiSP}. 

\vh
\noindent
{\bf Path Extensions.} We use two different path extension methods, one for $k$-APSiSP and the other for $k$-All-SiSP.
Path extensions have been used before in the hidden paths algorithm for APSP~\cite{KKP93} and more recently, for
fully dynamic APSP~\cite{DI04}. These two path extension methods differ from each other, as noted in \cite{DI06}. 
Our path extension method for $k$-All-SiSP is inspired by the method in~\cite{DI04} to compute `locally shortest paths'.
However, our path extension method for $k$-APSiSP is not related to the earlier results, except for the fact that all 
path extension methods place suitable paths on a priority queue and extract paths of minimum weight.

\vone
3. {\it Computing $k$ simple shortest cycles through a single vertex ($k$-SiSC) and $k$ simple shortest cycles
through  every vertex ($k$-AVSiSC).}   We show reductions between $k$-SiSC and 
$k$-SiSP, which give both algorithms and hardness results for $k$-SiSC. For $k$-AVSiSC, we 
give an $O(mn + n^2 \log n)$ time algorithm for $k=2$ using our 2-APSiSP algorithm, and
an algorithm that performs $n$  $k$-SiSP computations for $k>2$.

\vone
\noindent
{\bf Our Major Theorems.}
Here are the main theorems we establish for our algorithmic results. Some conditional hardness results through reductions are presented in Section~\ref{sec:hardness} for path and cycle problems on sparse graphs (as shown in 
Table~\ref{table1}).
  In all cases, the input is a directed graph $G=(V,E)$ with nonnegative
edge weights, and $|V|=n$, $|E|=m$.

\begin{theorem}\label{thm1}
Given an integer $k>1$, and the nearly simple shortest paths sets $Q_k(x,y)$ (Definition~\ref{def:qk}) for all $x,y\in V$,  Algorithm
{\sc Compute-APSiSP} (Section~\ref{sec:compute-pstar}) produces the $k$ simple shortest paths for every pair of vertices in 
$O(k \cdot n^2 + n^2 \cdot\log n)$ time.
\end{theorem}

\begin{theorem}\label{thm2}
(i) Algorithm {\sc 2-APSiSP} (Section~\ref{sec:2APSiSP}) correctly computes 2-APSiSP in $O(mn + n^2 \log n)$ time,
and for $k>2$, Algorithm {\sc APSiSP} (Section~\ref{sec:k-apsisp}) correctly computes $k$-APSiSP.\\
(ii) Let $T(m,n,k)$ be the time bound for Algorithm {\sc APSiSP}. 
\\ 
Then, $T(m,n,k) \leq n \cdot T(m,n,~k-1) + O(mn + k \cdot n^2 + n^2\cdot  \log n)$.\\
(iii) $T(m,n,3)$, the time bound for algorithm {\sc APSiSP} for $k=3$, is $O(m \cdot n^2 + n^3 \cdot \log n)$.
\end{theorem}

\begin{theorem}\label{thm3}
(i) $k$-All-SiSC: After an initial start-up cost of $O(mn + n^2 \log n)$ time, we can compute each successive simple shortest cycle in $O(mn + n^2 \log\log n)$ time
(Section~\ref{sec:allsisp1}).

(ii)  $k$-All-SiSP: 
After an initial start-up cost of $O(m)$ time to generate the first path, Algorithm~{\sc All-SiSP} (Section~\ref{sec:allsisp2}) computes 
each succeeding simple shortest path with the following bounds:
amortized $O(k + \log n)$ time if $k=O(n)$ and  $O(n + \log k)$ time if $k= \Omega (n)$, or  worst-case 
 $O(k \cdot \log n)$ time if $k=O(n)$, and $O(n \cdot \log k)$ time if $k= \Omega (n)$.
\end{theorem}

For the most part, we only consider computing the weights of the paths. The actual paths can be maintained by
using pointers to sub-paths that omit the first or last edge on the path.

In terms of conditional hardness results, we show that 2-SiSP, 2-AVSiSC and $2$-All-SiSP are all
at least as hard as finding a minimum weight cycle. We also show that $k$-SiSP is equivalent in
complexity to $k$-SiSC. These results are presented in Section~\ref{sec:hardness}.

Table~\ref{table1} 
lists our main results. Together they give a fairly complete understanding of the 
fine-grained complexity of the various natural problems related to computing $k$ simple shortest
paths and cycles in a weighted graph, at least
for $k=2$, assuming that finding a minimum weight cycle in `sub-$mn$ time'
is hard. Of these, we highlight the following contributions:

\vspace{-.05in}
\begin{itemize}
 \item
 The algorithms for $k$-APSiSP (and especially for 2- and 3-APSiSP)
and for $k$-All-SiSP introduce the new technique of path extensions for this class of problems.

\item
We show that $k$-SiSP and $k$-SiSC are equivalent in complexity, but we provide a hardness 
result that shows that $k$-All-SiSC is harder than $k$-All-SiSP unless we can obtain a significantly faster method for Min-Wt-Cyc. It is nevertheless interesting that we can generate successive simple shortest cycles in 
$\tilde{O}(mn)$  time, given that the mere enumeration of simple cycles was a much-investigated classical topic until linear-time generation of successive cycles (in no particular order) was
given in~\cite{Joh75}. 

\item We connect the complexity of several problems related to finding  $k$ simple paths and $k$ simple cycles in sparse graphs to the complexity of computing a minimum weight cycle.
For the most part, 
previous hardness results  were only 
for dense graphs, and with respect to the presence of sub-cubic algorithms.

 \item We give related results for undirected graphs and for unweighted graphs in Section~\ref{sec:hardness}.
\end{itemize}

\section{The $k$-APSiSP Algorithm}\label{sec:apsisp}

In this section, we present our algorithm to compute  $k$-APSiSP on  a directed graph $G=(V,E)$ with nonnegative edge-weight function $wt$.
The 
algorithm has two main steps. In the first step it computes the nearly $k$-SiSP sets $Q_k (x,y)$ for all pairs  $x,y$. In the second step it computes the exact $k$-SiSP sets $P^*_k (x,y)$ for all $x,y$ using the $Q_k(x,y)$ sets. This second step is the same for any value of $k$, and we describe this step first
in
Section~\ref{sec:compute-pstar}. We then 
present efficient algorithms to compute the $Q_k$ sets for $k=2$ and $k>2$ in Section~\ref{sec:compute-qk}.

In all of our algorithms we will maintain the paths in each $P^{*}_{k}(x,y)$  and $Q_k(x,y)$ set in 
an array in nondecreasing order of edge-weights.

\subsection{The Compute-APSiSP Procedure}\label{sec:compute-pstar}
 
 In this section we
present an algorithm, {\sc Compute-APSiSP}, to compute $k$-APSiSP. 
 This algorithm takes as input, the graph $G$, together with the nearly $k$-SiSP sets $Q_k (x,y)$, for each pair of
 distinct vertices $x,y$, and outputs  the $k^*$ simple shortest paths from $x$ to $y$ in the set $P^{*}_{k}(x,y)$ 
 for each pair of vertices $x, y \in V$ (note that $k^*$, which is defined in Definition~\ref{def:qk}, can be different for different vertex pairs
$x,y$). 
 As noted above, the construction of the $Q_k(x,y)$ sets will be described in the next section.

 The {\it right (left) subpath} of a path $\pi$ is defined as the path obtained by removing the first
 (last) edge on $\pi$. If $\pi$ is a single edge $(x,y)$ then this path is the vertex $y$ ($x$).
 
 \begin{lemma} \label{lemma:rightSubpath}
Suppose there are $k$ simple shortest paths from $x$ to $y$, all having the same first edge $(x,a)$. Then $\forall i$, $1 \leq i \leq k$, the right subpath of the $i$-th simple shortest path from $x$ to $y$ has weight equal to the
weight of the  $i$-th simple shortest path from $a$ to $y$.
\end{lemma}
\vspace{-0.2in}
\begin{proof}
By induction on $k$. Since subpaths of shortest paths are shortest paths, the statement holds for $k=1$.
Assume the statement is true for all $h \leq k$, and consider the case when the $h+1$ simple shortest paths from
$x$ to $y$ all share the same first edge $(x,a)$. Inductively, the right subpath of each of the first $h$ simple shortest paths have the weight equal to the corresponding simple shortest paths from $a$ to $y$. Suppose the weight of the right subpath $\pi_{a,y}$ of the $(h+1)$-th simple shortest path from $x$ to $y$ is not equal to the weight of the $(h+1)$-th simple shortest path from $a$ to $y$. 
Hence, if $\pi^{'}_{a,y}$ is the $(h+1)$-th simple shortest path from $a$ to $y$, we must have $wt(\pi_{a,y}) > wt(\pi^{'}_{a,y})$.

Since  $\pi_{xa,y}$ is the $(h+1)$-th simple shortest path from $x$ to $y$ and $wt(\pi_{a,y}) > wt(\pi^{'}_{a,y})$, there exists at least one path from $a$ to $y$ that contains $x$ and is also the $j$-th simple shortest path from $a$ to $y$, where $j \leq h+1$. Let this path be $\pi^{''}_{a,y}$. Let the subpath of $\pi^{''}_{a,y}$ from $x$ to $y$ be $\pi^{''}_{xa',y}$. But then $wt(\pi^{''}_{xa',y}) < wt(\pi^{''}_{a,y}) \leq wt(\pi^{'}_{a,y}) < wt(\pi_{a,y}) < wt(\pi_{xa,y})$. But this is a contradiction to our assumption that all the first $h+1$ simple shortest paths from $x$ to $y$ contains $(x,a)$ as the first edge. 
This contradiction establishes the induction step and the lemma.
\end{proof}

 Algorithm {\sc Compute-APSiSP} computes the $P^{*}_k(x,y)$ sets 
by extending an existing path by an edge. In particular, if the $k$-SiSPs from $x$ to $y$ all use the same first edge $(x,a)$, then it computes the $k$-th SiSP  by extending the $k$-th SiSP from $a$ to $y$  (otherwise, the sets $P^{*}_k(x,y)$ 
 are trivially computed from the sets $Q_k (x,y)$).
The algorithm first initializes the $P^*_k (x,y)$ sets with the corresponding $Q_k(x,y)$ sets
in Step \ref{alg2:addFromQk}. 
In Step \ref{alg2:ifinit}, it checks whether the shortest $k-1$ paths in $P^{*}_k(x,y)$ have the same first edge and if so,
by definition of $Q_k(x,y)$, this  $P^{*}_k(x,y)$ may not have been correctly initialized, and
may need to
update  
its
$k$-th shortest path to obtain the correct output. 
In this case, the common first edge $(x,a)$ is added to the set $\MarkEdges(a,y)$ in Step \ref{alg2:addToMarkEdges}. We explain this step below.

\begin{algorithm}[H]
\scriptsize
\begin{algorithmic}[1]
\State \underline{Initialize:}	\label{alg2:init}
\State $H \leftarrow \phi$ $~~~$ \{$H$ is a priority queue.\}
\ForAll {$x,y \in V, x \neq y$}	\label{alg2:for1}
 \State   $P^{*}_{k}(x,y) \leftarrow Q_k(x,y)$	\label{alg2:addFromQk}
 \If{the $k-1$ shortest paths in $P^{*}_{k}(x,y)$ have the same first edge }\label{alg2:ifinit}
\State Let $(x,a)$ be the common first edge in the $(k-1)$ shortest paths in $P^{*}_{k}(x,y)$	\label{alg2:assignCommon}
\State Add $(x,a)$ to the set $\MarkEdges(a,y)$	\label{alg2:addToMarkEdges}
\If{$|Q_k(a,y)| = k$}	\label{alg2:if1}
\State $\pi \leftarrow$ the path of largest weight in $Q_{k}(a,y)$
\State $\pi' \leftarrow (x,a) \circ \pi$
\State Add $\pi'$ to $H$ with weight $wt(x,a) + wt(\pi)$	\label{alg2:addToH1}
\EndIf	\label{alg2:endIf1}
\EndIf	\label{alg2:endFor2}
\EndFor\label{alg2:endFor1}
\State \underline{Main Loop:}
\While{$H \neq \phi$}	\label{alg2:while}
 \State $\pi \leftarrow$ {\sc Extract-min}$(H)$	\label{alg2:extractMin}
 \State Let $\pi = (xa, y)$ and let the path of largest weight in $P^{*}_{k}(x,y)$ be $\pi'$	\label{alg2:letPi}
 \State  {\bf if} {$|P^{*}_k(x,y)| = k-1$} {\bf then} 	add $\pi$ to $P^{*}_{k}(x,y)$ and set {\it update} flag \label{alg2:if2} \label{alg2:addToP1}
\State {\bf else} {\bf if} $wt(\pi) < wt (\pi')$ {\bf then} replace $\pi'$ with $\pi$ in $P^{*}_{k}(x,y)$ and set {\it update} flag	\label{alg2:addToP2}
 \If{{\it update} flag is set}
 	\State {\bf for all}  {$(x',x) \in \MarkEdges(x,y)$} {\bf do}  add $(x',x) \circ \pi$ to $H$ with weight $wt(x',x) + wt(\pi)$	\label{alg2:addToH2}
 \EndIf 
  \EndWhile	\label{alg2:endWhile}
\end{algorithmic}
\caption{{\sc Compute-APSiSP}$(G=(V,E), wt, k, \{Q_k(x,y), \forall x,y\})$}
\end{algorithm}\label{alg:compute-apsisp}

We define the {\it $k$-Left Extended Simple Path ($k$-LESiP) $\pi_{xa,y}$}  from $x$ to $y$ as the path
$\pi_{xa,y} = (x,a) \circ \pi_{a,y}$, where the path $\pi_{a,y}$ is the $k$-th shortest path
in $Q_k (a,y)$, and $\circ$ denotes the concatenation operation. In our algorithm we will construct 
$k$-LESiPs
for those pairs $x,y$ for which the $k-1$ simple shortest paths all start with the edge $(x,a)$.
The algorithm also maintains a set $\MarkEdges(a,y)$ for each pair of distinct vertices $a,y$; this set
contains those edges $(x,a)$ incoming to $a$ which are the first edge on all $k-1$ SiSPs from $x$ to $y$.
In addition to adding
the common first edge $(x,a)$ in the  $(k-1)$ SiSPs in $P^{*}_k(x,y)$ to   $\MarkEdges(a,y)$ in Step \ref{alg2:addToMarkEdges}, the algorithm creates 
the $k$-LESiP with start edge $(x,a)$ and end vertex $y$ 
using the $k$-th shortest path in the set $P^{*}_k(a,y)$,
and adds it to heap $H$ in  Steps~\ref{alg2:if1}-\ref{alg2:endIf1}. Let $\mathcal{U}$ denote the set
of $P^*_k (x,y)$ sets which may need to be updated; these are the sets for which the if 
condition in Step~\ref{alg2:ifinit} holds.

In the main while loop in Steps \ref{alg2:while}-\ref{alg2:endWhile}, a min-weight path is extracted in each iteration. We establish
below that  this min-weight path is added to the corresponding   $P^{*}_k$ in Step \ref{alg2:addToP1} or \ref{alg2:addToP2} only
if it is the $k$-th SiSP; in this case, its left extensions are created and added to the heap $H$ in Step \ref{alg2:addToH2}.

\begin{lemma}	\label{lemma:computeAPSiSP}
Let $G=(V,E)$ be a directed graph with nonnegative edge weight function $wt$, and $\forall x,y \in V$, 
let the set $Q_{k}(x,y)$ contain the nearly $k$-SiSPs from $x$ to $y$. 
Then,
algorithm {\sc Compute-APSiSP} correctly computes the sets $P^{*}_{k}(x,y)$ $\forall x,y \in V$.
\end{lemma}

\vspace{-0.2in}
\noindent

\begin{proof}
First, we need to show that the paths in sets $P^{*}_{k}(x,y)$ are indeed simple. Clearly, the paths added to $P^{*}_k$ from sets $Q_{k}$ in Step \ref{alg2:addFromQk} are already simple (from the definition of $Q_k$). So we only need to show that the paths
 added to $P^{*}_{k}$ in Steps  \ref{alg2:addToP1} and \ref{alg2:addToP2} are
  simple. To the contrary assume that some of the paths that are added to $P^{*}_{k}$ are non-simple. Clearly these paths must be of length greater than 1. Let $\pi_{xa,y} = x \rightarrow a \rightsquigarrow y$ be the first minimum weight path extracted from
  $H$  that contains a cycle and was added to $P^{*}_{k}$ in Step \ref{alg2:addToP1} or \ref{alg2:addToP2}. Clearly, $P^{*}_k(x,y) \in \mathcal{U}$ and $(x,a) \in \MarkEdges(a,y)$ and the right subpath $\pi_{a,y}$ must be in $P^{*}_{k}$ (otherwise the path $\pi_{xa,y}$ would never have been added to heap $H$ in Step \ref{alg2:addToH1} or \ref{alg2:addToH2}). The right subpath $\pi_{a,y}$ must also be simple (as $wt(\pi_{a,y}) < wt(\pi_{xa,y})$), and it must contain $x$ in order to create a cycle in $\pi_{xa,y}$.
   Let $\pi_{xa',y}$ ($a' \neq a$) be the subpath of $\pi_{a,y}$ from $x$ to $y$. Now there are two cases depending on whether $\pi_{xa,y}$ was added to $P^{*}_k$ in Step \ref{alg2:addToP1} or \ref{alg2:addToP2}.

If $\pi_{xa,y}$ was added to $P^{*}_k(x,y)$ in Step \ref{alg2:addToP1} and as $P^{*}_k(x,y) \in \mathcal{U}$ , it implies that all $k-1$ paths in $Q_k(x,y)$ have same first edge $(x,a)$ and there is no simple path from $x$ to $y$ 
in $Q_k(x,y)$
 with some first edge $(x,a'') \neq (x,a)$. This is a contradiction as the subpath $\pi_{xa',y}$ of $\pi_{a,y}$ contains $(x,a') \neq (x,a)$ as its first edge. 

Otherwise, let $\pi_{xa'',y} \in Q_k(x,y)$ ($a'' \neq a$) be the path
 that was removed from $P^{*}_k$ in Step \ref{alg2:addToP2} to accommodate $\pi_{xa,y}$. Thus, we have $ wt(\pi_{xa',y}) < wt(\pi_{xa,y}) < wt(\pi_{xa'',y}) $, which is a contradiction as $\pi_{xa'',y} \in Q_k(x,y)$ and is the shortest path from $x$ to $y$ avoiding edge $(x,a)$ (as the other $k-1$ shortest paths in $Q_k(x,y)$ have $(x,a)$ as the first edge). As path $\pi_{xa,y}$ is arbitrary, hence all paths in $P^{*}_k$ are simple.

Now we need to show that $P^{*}_k (x,y)$ indeed contains the $k^*$ SiSPs from $x$ to $y$.

From the definition of $Q_k(x,y)$, it is evident that $P^{*}_k(x,y)$ indeed contains the $k-1$ SiSPs from $x$ to $y$. 
We now need to show that the $k$-th shortest path in each of the sets $P^{*}_k$ is indeed the corresponding $k$-th SiSP. To the contrary assume that there exists a $P^{*}_k$ set that does not contain the correct $k$-th SiSP. Let $\pi_{xa,y} = x \rightarrow a \rightsquigarrow y$ be the minimum weight $k$-th SiSP that is not present in $P^{*}_k$. Clearly, $\pi_{xa,y} \notin Q_k(x,y)$ (otherwise it would have been added to $P^{*}_k(x,y)$ in Step \ref{alg2:addFromQk}). 
This implies that $\pi_{xa,y}$ has the same first edge as that of the $k-1$ SiSPs from $x$ to $y$ and hence $P^{*}_{k}(x,y) \in \mathcal{U}$ and $(x,a) \in \MarkEdges(a,y)$. By 
Lemma \ref{lemma:rightSubpath}, the right subpath of $\pi_{xa,y}$ must have weight equal to the $k$-th SiSP from $a$ to $y$. Thus, there are at least $k$ SiSPs from $a$ to $y$ and 
the set $P^{*}_k(a,y)$ contains all the $k$ SiSPs from $a$ to $y$. And as $(x,a) \in \MarkEdges(a,y)$, a path $\pi^{'}_{xa,y}$ with the $k$-th SiSP from $a$ to $y$ as the right subpath and weight equal to $wt(\pi_{xa,y})$ must have been added to $H$ either in Step \ref{alg2:addToH1} or \ref{alg2:addToH2} and would have been added to $P^{*}_k(x,y)$ in Step \ref{alg2:addToP1} or \ref{alg2:addToP2}, resulting in a contradiction to our assumption that $P^{*}_k(x,y)$ does not contain all the $k$ SiSPs. Thus, $P^{*}_k (x,y)$  does contain the $k^*$ SiSPs from $x$ to $y$. 
\end{proof}

The time bound for Algorithm {\sc Compute-APSiSP}  in Theorem~\ref{thm1}
is established with the following sequence of simple lemmas.

\begin{lemma}\label{lemma:markEdges}
There are $O(kn^2)$ paths in $P^{*}_k$, and $O(n^2)$ elements across all  \MarkEdges{} sets.
\end{lemma}
\vspace{-0.2in}

\begin{proof}
$|P^*_k (x,y)| = O(kn^2)$ since 
 there are at most $k$ paths in each of the $n \cdot (n-1)$ sets  $P^{*}_{k}(x,y)$. For the second part, exactly one
edge is contributed to a \MarkEdges{} set by each $P^*_k (x,y) \in \mathcal U$ in 
Step~\ref{alg2:addToMarkEdges}.
\end{proof}

\begin{lemma}	\label{lemma:update}
 Each $P^{*}_{k}(x,y)$ set is updated at most once
in the main while loop.
\end{lemma}
\vspace{-0.2in}

\begin{proof}
 A path can be added to $P^*_k(x,y)$
 at most once in Step~\ref{alg2:addToP1} since its
size will increase to $k$ after the addition. Also, a path is added at most once in either
Step~\ref{alg2:addToP1} or Step~\ref{alg2:addToP2} since paths are extracted from $H$ in
nondecreasing order of their weights.
\end{proof}

\begin{lemma}	\label{lemma:LESiP}
The number of  $k$-LESiPs added to heap $H$ is $O(n^2)$.
\end{lemma}
\vspace{-0.2in}

\begin{proof}
For each $k$-LESiP, the right subpath must be the $k$-th shortest path in $P^{*}_k$. For each pair of vertices $x,y \in V$, there is at most one entry across the \MarkEdges{}  sets (say edge $(x,a) \in \MarkEdges(a,y)$) and hence at most one $k$-LESiP will be added to heap $H$ in Step \ref{alg2:addToH1} for pair $(x,y)$. By lemma \ref{lemma:update}, we know that the set $P^{*}_k(a,y)$ is updated at most once and hence at most one $k$-LESiP will be added to heap $H$ for pair $(x,y)$ in Step \ref{alg2:addToH2}. Thus, there are only $O(n^2)$ $k$-LESiPs that were added to the heap $H$ in the algorithm.
\end{proof}

\vmone
\begin{lemma}	\label{lemma:runtimeComputeAPSiSP}
Algorithm {\sc Compute-APSiSP} runs in $O(kn^2 + n^2\log n)$ time.
\end{lemma}

\vspace{-.2in}
\begin{proof}
A binary heap suffices for $H$.
The initialization for loop in Steps \ref{alg2:for1}-\ref{alg2:endFor1} takes $O(kn^2)$ time to initialize and inspect the $P^*_k$ sets. It is executed at most $n^2$ times and, outside of the 
inspection of $P^*_k(x,y)$ an  iteration costs $\Theta(\log n)$ time (cost for insertion in heap), thus contributing $O(n^2\log n)$ to the running time. The while loop is executed $O(n^2)$ times as by lemma \ref{lemma:LESiP}, $O(n^2)$ elements are added to the heap. The extract-min operation takes $\Theta(\log n)$ time and hence Step \ref{alg2:extractMin} contributes $O(n^2\log n)$ to the running time. Steps \ref{alg2:letPi}-\ref{alg2:addToP2} takes constant time per iteration and hence add $O(n^2)$ to the total running time. By lemma \ref{lemma:markEdges}, Step \ref{alg2:addToH2} is executed $O(n^2)$ times and contributes $O(n^2\log n)$ to the running time. Thus, the total running time of the algorithm is $O(kn^2 + n^2\log n)$.
\end{proof}

\subsection{Computing the $Q_k$ Sets}\label{sec:compute-qk}

\subsubsection{Computing $Q_k$ for $k=2$}	\label{sec:2APSiSP}

We now give an $O(mn + n^2 \log n)$ time algorithm to
compute $Q_{2}(x,y)$  for all pairs $x,y$. We then show that we can also obtain the $Q_2$ sets from a DSO (distance sensitivity oracles, see Introduction), but this algorithm is slightly slower than our first method.  

Our faster method first computes a  shortest path (SP) for each pair using an efficient 
APSP algorithm~\cite{Pettie04}.  
This gives the first path in each $Q_2$ set. To obtain the second path, for each $x,y$ we need to find a
shortest path from $x$ to $y$ that avoids first edge $(x,a)$ on the SP.
We can trivially compute such paths by running Dijkstra on the subgraph $G - \{e\}$ with source $x$ where $e = (x,a)$ is the first edge on the shortest path from $x$ to $y$. With this approach we will make  $m$ calls to Dijkstra's algorithm.
 We now describe a more efficient method that makes only $n$ calls to Dijkstra's algorithm. This method uses
  the procedure {\sc fast-exclude} from Demetrescu et al.~\cite{DTCR08}. We present the
  input-output specifications of {\sc fast-exclude} here; full details of this algorithm can be found 
  in~\cite{DTCR08}. We start with the following
  definition.
 
\begin{definition} [\textit{Independent Edges} {\cite{DTCR08}}]
Given a rooted tree $T$, edges $(u_1,v_1)$ and $(u_2,v_2)$ on $T$
are \emph{independent} if the subtree of $T$ rooted at $v_1$ and the subtree of $T$ rooted at $v_2$ are disjoint.
\end{definition}

Given the weighted directed graph $G=(V,E)$,  the SSSP tree $T_s$ rooted
at a source vertex $s\in V$, and 
a set $S$ of independent edges in $T_s$,
algorithm {\sc fast-exclude} in~\cite{DTCR08} computes, for each edge $e \in S$, a shortest path 
from $s$ to every other vertex in $G - \{e\}$. This algorithm runs in time $O(m + n \log n)$. 

We will compute the second path in each $Q_2 (x,y)$ set, for
a given $x\in V$, by running
{\sc fast-exclude} with $x$ as source, and with the set of  outgoing edges from $x$ in
$T_x$ as the set $S$. 
Clearly,
this set $S$ is independent, and hence algorithm {\sc fast-exclude} will produce
its specified output.
Now consider any vertex $y \neq x$, and let $(x,a)$ be the first edge on the shortest path from $x$ to $y$ in $T_x$. Then, by its specification, {\sc fast-exclude} will
compute a shortest path from $x$ to $y$ that avoids edge $(x,a)$ in its output, which is the second
path needed for $Q_2(x,y)$. This holds for every vertex $y\in V-\{x\}$. Thus 
we have the following:

\begin{lemma}\label{lem:q2}
The sets $Q_2(x,y)$, for all pairs $x,y$, can be computed in $O(mn + n^2 \log n)$ time.
\end{lemma}

This leads to the following algorithm for 2-APSiSP. Its time bound in Theorem~\ref{thm2}, part $(i)$  follows from Lemmas \ref{lem:q2}, \ref{lemma:computeAPSiSP} and \ref{lemma:runtimeComputeAPSiSP}.

\begin{algorithm}[H]
\scriptsize
\begin{algorithmic}[1]

\For {each $x \in V$}	\label{alg3:for1}
  \State Compute the shortest path in each $Q_2 (x,y)$, $y \in V-\{x\}$, by running Dijkstra's algorithm with source $x$.
	\State Compute the second path in each $Q_2(x,y)$, $y \in V-\{x\}$,  using  {\sc fast-exclude} with source $x$ and $S=\{(x,a)\in T_x\}$  \label{alg3:computeQ2}
\EndFor	\label{alg3:endFor1} 
\State {\sc Compute-APSiSP}($G$, wt, 2,  $\{Q_2(x,y), \forall x,y\}$)	\label{alg3:callComputeAPSiSP}
\end{algorithmic}
\caption{{\sc 2-APSiSP}$(G=(V,E); wt)$}
\end{algorithm}

{\bf Computing the $Q_2$ sets from distance sensitivity oracle.}  Let a DSO $D$ with constant query time be given.
For each $x, y \in V$, let $\pi_{xy}$ be the shortest path from $x$ to $y$. 
The second 
SiSP in $Q_2(x,y)$ is the shortest path from $x$ to $y$ avoiding the first edge on $\pi_{xy}$, so we can compute the second SiSP in $Q_2(x,y)$ by making $O(1)$ queries to $D$. Thus,  $O(n^2)$ queries suffice to compute 
the second SiSP in all $Q_2(x,y)$ sets. 
A DSO with constant query time can be computed by a randomized algorithm in $O(n \log n \cdot (m + n \log n ))$ time,
and deterministically in $O(n \log^2 n \cdot (m + n \log n))$ time~\cite{BK09}. Since {\sc Compute-APSiSP} runs in $O(n^2 \log n)$, this gives a
$\tilde{O}(mn)$ time algorithm for 2-APSiSP.
  It is not clear if we can efficiently compute 2-APSiSP directly from a DSO, without using the $Q_2$ sets
  and {\sc Compute-APSiSP}.

\subsubsection{The Algorithm for $k\geq 3$}\label{sec:k-apsisp}

Our algorithm will use the following types of sets. For each vertex $x \in V$, let $I_x$ be the set of incoming edges to
$x$. Also, for a vertex $x \in V$, and vertices $a, y \in V - \{x\}$, let $P^{*x}_k(a,y)$ be the set of $k$ simple shortest
paths from $a$ to $y$ in $G - I_x$, the graph obtained after removing the incoming edges to $x$.
Recall that we maintain all $P^*$ and $Q$ sets as sorted arrays.

We now present Algorithm {\sc APSiSP}$(G, k)$, which first computes the sets $P^{*x}_{k-1}(a,y)$, for all vertices $a, y \in V$.
Once we have these sets, each $Q_k(x,y)$ can be computed as the set of all  paths in the set $P^{*}_{k-1}(x,y)$, together with a shortest path in $\bigcup_{\{(x,a) \mbox{ outgoing from $x$}\}} \{(x,a)\circ p ~|~ p \in P^{*x}_{k-1}(a,y)\}$ (which
is not present in $P^{*}_{k-1}(x,y)$).

\begin{algorithm}[H]
\scriptsize
\begin{algorithmic}[1]
\If {$k = 2$}
	\State compute $Q_2$ sets using algorithm in Section~\ref{sec:2APSiSP}
\Else
	\For {each $x \in V$}	\label{alg4:for1}
		\State $I_x \leftarrow$ set of incoming edges to $x$
		\State Compute sets $P^{*x}_{k-1}(x,y)$, and $P^{*x}_{k-1}(a,y)$ $\forall a,y \in V$ by calling {\sc APSiSP}$(G-I_x, wt, k-1)$	\label{alg4:computeQ3}
		\For {each $y \in V-\{x\}$}	\label{alg4:for2}
			\State $Q_k(x,y) \leftarrow P^{*x}_{k-1}(x,y)$	\label{alg4:Qkstart}
			\State {\bf for all} $(x,a) \in E$ {\bf do} $count_a \leftarrow$ number of paths in $Q_k(x,y)$ with $(x,a)$ as the first edge
			\State $Q_k(x,y) \leftarrow Q_k(x,y) \cup \{$ a shortest path in $\bigcup_{\{(x,a) \mbox{ outgoing from $x$}\}} (x,a)\circ P^{*x}_{k-1}(a,y)[count_a +1]\}$	\label{alg4:Qkend}
	 	\EndFor
	\EndFor	\label{alg4:endFor1} 
\EndIf
\State {\sc Compute-APSiSP}$(G, wt, k, \{Q_k(x,y)~ \forall x,y\in V\})$	\label{alg4:callComputeAPSiSP}
\end{algorithmic}
\caption{{\sc APSiSP}$(G=(V,E),~wt,~k)$}
\end{algorithm}

To compute the $P^{*x}_{k-1}$ sets, {\sc APSiSP}$(G,wt,k)$ recursively calls {\sc APSiSP}$(G - I_x,wt,k-1)$, for each 
vertex $x \in V$. Once we have computed the $P^{*x}_{k-1}$ sets, the $Q_k(x,y)$ sets are readily computed as 
described in steps \ref{alg4:Qkstart} - \ref{alg4:Qkend}. After the computation of $Q_k(x,y)$ sets, 
{\sc APSiSP}$(G,wt,k)$ calls {\sc Compute-APSiSP}$(G, wt, k, \{Q_k(x,y)~ \forall x,y\in V\})$ to compute the $P^{*}_k$
sets.  This establishes the following 
lemma and part $(ii)$ of Theorem~\ref{thm2}.

\begin{lemma}	\label{lemma:3APSiSP}
Algorithm {\sc APSiSP}
$(G,wt,k)$
 correctly computes the sets $P^{*}_{k}(x,y)$ $\forall x,y \in V$. 
\end{lemma}

{\it Proof of Theorem~\ref{thm2}, part $(iii)$.}
The for loop starting in Step \ref{alg4:for1} is executed $n$ times, and the cost of each iteration
is dominated by the call to Algorithm  {\sc 2-APSiSP} in Step \ref{alg4:computeQ3}, which
 takes $O(mn + n^2\log n)$ time. This contributes $O(mn^2 + n^3\log n)$ to the total running time.
The inner for loop starting in Step \ref{alg4:for2} is executed $n$ times per iteration of the outer for loop, and the 
cost of each iteration is $O(k + d_x)$. Summing over all $x \in V$, this contributes $O(kn^2 + mn)$ to the
total running time.
  Step \ref{alg4:callComputeAPSiSP} runs in $O(n^2\log n)$ time as shown in Lemma \ref{lemma:runtimeComputeAPSiSP}. Thus, the total running time is $O(mn^2+ n^3\log n)$.
\hfill $\square$

{\bf $k$-APSiSP.} The performance of Algorithm {\sc APSiSP}  degrades by a factor
of $n$ with each increase in $k$. Thus, it matches Yen's algorithm (applied to all-pairs) for $k=4$, and for larger values of
$k$ its performance is worse than Yen. 

Since finding the $P^*_k$ sets is at least as hard as finding the $Q_k$ sets (as long as the running
time is $\Omega (k \cdot n^2 + n^2 \log n)$), it is possible that the for loop starting in Step~\ref{alg4:for1}
could be replaced by a faster algorithm for finding the $Q_k$ sets, which in turn would lead
to a faster algorithm for $k$-APSiSP.

\subsection{Generating $k$ Simple Shortest Cycles}\label{sec:k-cycles}

\noindent
{\bf $k$-SiSC.}
This is
the problem of generating the $k$ simple shortest cycles through a specific vertex $z$ in $G$
($k$-SiSC). We can reduce this problem to $k$-SiSP by forming $G'_z$, where we replace vertex $z$
by vertices $z_i$ and $z_o$ in $G'_z$,  and we replace each incoming edge to (outgoing edge from) $z$ with an incoming edge to $z_i$ (outgoing edge
from $z_o$) in $G'_z$. It is not difficult to see that the $k$-th simple shortest path from $z_o$ to $z_i$ in $G'_z$ corresponds to the $k$-th simple
shortest cycle through $z$ in $G$.

\vh
\noindent
{\bf $k$-AVSiSC.}
This is
 the problem of generating $k$ simple shortest cycles that pass through a given vertex $x$, 
{\it for every vertex}  $x\in V$.
For $k=2$, we can reduce this problem to $k$-APSiSP by forming the graph $G'$ where for each vertex $x$,  we replace vertex $x$ in $G$ by vertices $x_i$ and $x_o$ in $G'$, we place a directed
edge of weight 0 from $x_i$ to $x_o$, and we replace each edge $(u,x)$ in $G$ by an edge $(u_o,x_i)$ in $G'$ (and hence we also replace
each edge $(x,v)$ in $G$ by an edge $(x_o, v_i)$ in $G'$). For $k>2$ a faster algorithm would repeat
$k$-SiSC for each vertex.
This leads to the following theorem.

\begin {theorem}
Let $G$ be a directed graph with non-negative edge weights.  Then, \\
(i) $k$-SiSC can be computed in $O(k \cdot (mn + n^2 \log\log n))$ time, the same time as $k$-SiSP.\\
(ii) $2$-AVSiSC can be computed in $O(mn + n^2 \log n)$ time, and for $k>2$, $k$-AVSiSC can be 
computed in $O(k \cdot n \cdot (mn +n^2 \log\log n))$ time.
\end{theorem}

\section{Enumerating Simple Shortest Paths and Cycles in a Graph}\label{sec:allsisp}
 
 In this section we consider the problem of successively generating simple paths and cycles in nondecreasing order of their weights in a directed
 $n$-node, $m$-edge graph $G=(V,E)$ with nonnegative edge weights. 
 In Section~\ref{sec:allsisp1} we give a method to generate each successive simple shortest cycle ($k$-All-SiSC) in $\tilde{O}(m \cdot n)$ time. For enumerating simple paths in nondecreasing order of weight ($k$-All-SiSP), we give a faster method
 in Section~\ref{sec:allsisp2} that  uses again a path extension method, different from the one used in Section~\ref{sec:compute-pstar}. On the other hand, in 
 Section~\ref{sec:hardness} we show 
 that the problem of generating the $k$-th simple shortest cycle in a graph after the first $k-1$ cycles have been generated is at least as hard as the  Min-Wt-Cyc problem.
 
\subsection{Generating Successive Simple Shortest Cycles}\label{sec:allsisp1}

We assume the vertices are numbered 1 through $n$. 
Our algorithm for 
$k$-All-SiSC 
 maintains an array $A[1..n]$, where each $A[j]$  contains a triple $(ptr_j, w_j, k_j)$; here
$ptr_j$ is a pointer to the shortest  cycle, not yet generated, that contains  $j$ as the minimum vertex  (if such a cycle exists),
$w_j$ is the weight of this cycle, and $k_j$ is the number of  
shortest simple cycles through vertex $j$ that have already been generated.
(Note that if a cycle $C$ is pointed to by an entry  in $A[r]$, then the minimum vertex on $C$ must be labelled $r$; thus any given cycle is assigned to exactly one position in array $A$.)

We will work with the graph $G'$ described in Section~\ref{sec:k-cycles}. Initially, 
we compute the entry for  each $A[j]$ by running Dijkstra's algorithm with source $j_o$ on the subgraph $G'_j$ of $G'$  induced on 
$V'_j = \{x_i, x_o ~|~ x \geq j\}$, to find a shortest path $p$ from $j_o$ to $j_i$; we then initialize $A[j]$ with a pointer to the cycle in $G$ associated with $p$, and with its weight, and with $k_j=0$.

For each $k\geq 1$, we generate the $k$-th simple shortest cycle in $G$ 
by choosing  a minimum weight cycle in array $A$. Let this entry be in $A[r]$.
We then compute the last path in $(k_r +1)$-SiSC  through vertex $r$ using the algorithm in Section~\ref{sec:k-cycles},
and we update the entry  in $A[r]$ with this cycle.

Correctness of this algorithm is immediate since the $k$-th simple shortest cycle must be pointed to by some entry in array $A$ after $k-1$ iterations.
The initialization takes $O(mn + n^2 \log n)$ for the $n$ calls to Dijkstra's algorithm. Thereafter,
the algorithm in Section~\ref{sec:k-cycles} generates each new cycle in the slightly faster APSP time bound
of $O(mn + n^2 \log\log n)$,  by maintaining the relevant information generated during the computation  of earlier cycles, as in~\cite{Yen71,GL09}. This establishes the correctness of Theorem~\ref{thm3}, part $(i)$.

A similar algorithm can generate successive simple shortest paths. But in the next section, we present a faster algorithm for this
problem. For constant $k$, this algorithm generates a succinct representation of the $k$-th simple shortest path in $O(\log n)$
time, after an initial start-up cost of $O(m)$ to generate a shortest simple path in the graph (which is an edge of minimum weight).

 \subsection{A Faster Algorithm to Generate Successive Simple Shortest Paths}\label{sec:allsisp2}
 
 Since all vertices on a simple path must be distinct, an $n$ node graph has $O(n^n)$ simple paths.
  Our algorithm for $k$-All-SiSP
   is inspired by  the method in~\cite{DI04} for fully dynamic APSP.

  With each path $\pi$, we will associate two sets of paths $L(\pi)$ and  $R(\pi)$ as described below. Similar sets
  are used in~\cite{DI04} for `locally shortest paths' but here they have a different use
 as described below.
   
   {\it Left and right extensions.} 
   Let $\mathcal P$ be a collection of simple paths.
   For a simple path $\pi_{xy}$ from $x$ to $y$ in $\mathcal P$,
    its left extension set $L(\pi_{xy})$
  is the set of simple  paths $\pi' \in \mathcal P$ such that $\pi' = (x', x) \circ \pi_{xy}$, for some $x' \in V$.
  Similarly, the right extension set $R(\pi_{xy})$ is the set of simple paths $\pi''= \pi_{xy} \circ (y,y')$ such that $\pi'' \in \mathcal P$. For a trivial path $\pi= \langle v \rangle$, 
  $L(\pi)$ is the set of incoming edges to $v$, and $R(\pi)$ is the set of outgoing edges from $v$.
  
  Algorithm {\sc All-SiSP}, given below,  generates all simple shortest paths in $G$ in nondecreasing order of
  weight. 
   To generate the $k$ shortest simple paths in $G$, we can  terminate
  the while loop after $k$ iterations.
   Algorithm {\sc All-SiSP}  initializes a priority queue $H$ with the edges in $G$, and it initializes the extension sets for
  the vertices in $G$. In each iteration of the main loop, the algorithm extracts the minimum weight path $\pi$ in $H$ as the next simple path in the output sequence. It then generates suitable extensions of $\pi$ to be added to $H$ as follows. Let the first edge
  on $\pi$ be $(x,a)$ and the last edge $(b,y)$. Then,  {\sc All-SiSP} left extends $\pi$
along those edges $(x',x)$ such that there is a path $\pi_{x'b}$ in $L(l(\pi))$; it also requires that
 $x' \neq y$, since extending to $x'$ would create a cycle in the path. Algorithm {\sc All-SiSP}  forms
similar extensions to the right in the for loop starting at Step~\ref{algoall:rightext}.

 \begin{algorithm}[H]
 \scriptsize
\begin{algorithmic}[1]
\State \underline{Initialization:}	\label{algall:init}
\State $H \leftarrow \phi$ $~~~$ \{$H$ is a priority queue.\}
\ForAll {$(x,y) \in E$} 
\State Add $(x,y)$ to priority queue $H$ with $wt(x,y)$ as key\label{algoall:init}
\State Add $(x,y)$ to $L(\langle y\rangle)$ and $R(\langle x\rangle)$
\EndFor
\State \underline{Main loop:}
\While{$H \neq \phi$}
\State $\pi \leftarrow$ {\sc Extract-min}$(H)$ \label{algoall:extract}
\State Add $\pi$ to the output sequence of simple paths
\State Let   $\pi_{xb}=\ell(\pi)$ and $\pi_{ay}= r(\pi)$ (so $(x,a)$ and $(b,y)$ are the first and last edges on $\pi$)
\ForAll {$\pi_{x'b} \in L(\pi_{xb})$ with $x' \neq y$} \label{algoall:forall1}
\State Form $\pi_{x'y} \leftarrow (x',x) \circ \pi$ and add $\pi_{x'y}$ to $H$ with $wt(\pi_{x'y})$ as key \label{algoall:leftext}
\State Add  $\pi_{x'y}$ to $L(\pi_{xy})$ and to $R(\pi_{x'b})$ \label{algoall:leftextsets}
\EndFor
\State {\bf for all} $\pi_{ay'} \in R(\pi_{ay})$ with $y' \neq x$ {\bf do} perform steps complementary to Steps \ref{algoall:leftext} and \ref{algoall:leftextsets} \label{algoall:rightext}
\EndWhile
\end{algorithmic}
\caption{{\sc All-SiSP}$(G=(V,E); wt)$}	\label{alg:all}
\end{algorithm} 

We now establish that Algorithm {\sc All-SiSP} generates only simple paths, and that it generates every
simple path in $G$ in nondecreasing order of weight. 

 \begin{lemma}	\label{lem:nocycle}
 Every path generated by Algorithm {\sc All-SiSP}
 is a simple path.
 \end{lemma} 
  \vspace{-0.2in}

  \begin{proof}
  Since edge weights are nonnegative,
 the first path generated by Algorithm~\ref{alg:all} is a minimum weight edge inserted in Step~\ref{algoall:init}, which is
  a simple path. Assume the algorithm generates a path with a cycle, and let $\sigma$ be the first path extracted in 
  Step~\ref{algoall:extract} that contains a cycle.  Let $(x',a)$ and $(b,y)$ be
  the first and last edges on $\sigma$. Since $\sigma$ contains a cycle, it contains at least two edges so $(x',a)$ and
  $(b,y)$ are distinct edges.
  
  Consider the step when the non-simple path
   $\sigma$ is placed on $H$. This does not occur in Step~\ref{algoall:init} since $\sigma$ contains at
  least two edges. So $\sigma$ is placed on $H$ in some iteration of the while loop. Let $\pi$ be the path extracted from
  $H$ in this iteration; $\pi$ is a simple path by assumption since it was extracted from $H$ before $\sigma$.
   Then $\sigma$ is added to $H$ either as a left extension of $\pi$ (in Step~\ref{algoall:leftext}) or
  as a right extension of $\pi$ in a step complementary to Step~\ref{algoall:leftext} in the for loop in 
  Step~\ref{algoall:rightext}. 
  
  Consider the left extension case, and let $\sigma$ be formed when processing
  path $\pi_{x'b} \in L(l(\pi))$ with $x' \neq y$ in Step ~\ref{algoall:forall1}. 
   Thus $\sigma$ is formed as $(x',x) \circ \pi$ in Step~\ref{algoall:leftext}. 
   But  $(x',x) \circ \pi = (x',x) \circ \ell(\pi) \circ (b,y) = \pi_{x'b} \circ (b,y)$.
Since $\pi_{x'b} \in L(l(\pi))$, it was also placed in $H$ in either Step~\ref{algoall:init} or Step~\ref{algoall:leftext}.
And as $wt(\pi_{x'b}) < wt(\sigma)$, the path  $\pi_{x'b}$ is simple.
   Since 
  $\pi_{x'b}$ is simple, a cycle can be formed in $\sigma$ only if $x' = y$. But this is specifically forbidden in
 the condition in Step~\ref{algoall:forall1}. A similar argument applies to right extensions added to $H$ in 
 Step~\ref{algoall:rightext}. Hence $\sigma$ is a simple path, and Algorithm~\ref{alg:all} does not
 generate any path containing a cycle.   
 \end{proof}

 \begin{lemma}	\label{lemma:all}
  Algorithm {\sc All-SiSP} generates all simple paths in $G$ in nondecreasing order of their weights.
  \end{lemma}
  \vspace{-0.2in}
 \begin{proof}
  Clearly the algorithm correctly generates the minimum weight edge in $G$ as the minimum weight simple path
  in the output in the first iteration of the while loop. By Lemma~\ref{lem:nocycle} all generated paths are simple.
  Also, these simple paths are generated in nondecreasing order of weight since any path added to $H$ in Steps
  ~\ref{algoall:leftext} and \ref{algoall:rightext} has weight at least as large as the weights
   of the paths that have been extracted at that time, due
  to nonnegative  edge-weights. It remains to show that no simple path in $G$ is omitted in the sequence of simple paths
  generated.
  
  Suppose the algorithm fails to generate all simple shortest paths in $G$ and let $\pi$ be a simple path of smallest weight
  that is not generated by Algorithm~\ref{alg:all}. Let $\pi$ be a path with first edge $(x,a)$ and last edge $(b,y)$;
  $(x,a) \neq (b,y)$ since all single edge paths is added to $H$ in Step~\ref{algoall:init}, and will be extracted in a future
  iteration. Let $\pi_{ab}$ be the
  subpath of $\pi$ from $a$ to $b$.
  By assumption, the paths 
  $\pi_{xb} = \ell(\pi)$ and $\pi_{ay}=r(\pi)$ are placed in the output by Algorithm~\ref{alg:all} since they are simple paths with weight smaller than
  the weight of $\pi$. Without loss of generality assume that $\pi_{xb}$ was extracted from $H$ before $\pi_{ay}$. 
  
  Clearly, $\pi_{xb}$ was inserted in $H$  before $\pi_{ay}$ was extracted.
  In the
  iteration of the while loop when $\pi_{xb}$ was added to $H$, $\pi_{xb}$ was added to 
  $L(\pi_{ab})$ in Step~\ref{algoall:leftextsets} since $r(\pi_{xb}) = \pi_{ab}$. In the later iteration when $\pi_{ay}$ was extracted from $H$,
 the paths in $L(\ell(\pi_{ay}))$ are considered in Step~\ref{algoall:leftext}. But 
  $\ell(\pi_{ay}) = \pi_{ab}$. When the paths in $L(\ell(\pi_{ay})) =L( \pi_{ab})$ are considered in Step~\ref{algoall:forall1} 
  during the processing 
  of $\pi_{ay}$, the path
  $\pi_{xb}$ will be one of the paths processed, and in Step~\ref{algoall:leftext} the path $(x,a) \circ \pi_{ay} = \pi$ 
  will be formed
  and added to $H$. Thus $\pi$ will be added to $H$, and hence will be extracted and added to the output sequence.
  \end{proof}

We can now prove  Theorem~\ref{thm3}.

{\it Proof of Theorem~\ref{thm3}, part $(ii)$.}
We will maintain paths with pointers to their left and right subpaths, so each path takes $O(1)$ space.
For the amortized bound we will implement $H$ as a Fibonacci heap.
The initialization takes $O(m)$ time. 
Each $L$ and $R$ set can contain at most $n-2$ paths, and further, since  extensions are
formed only with paths already in $H$, each of these sets has size $\min \{k, n-2\}$.
 The $k$-th iteration  of the while loop
takes time $O(\log |H|)$ for the extract-min operation, and $O(\min \{k, n\})$ time for the processing of the $L$ and $R$ sets.
At the start of the $k$-th iteration, the number of paths in $H$ is at most $O(m + k \cdot \min \{k, n\})$, and since $m = O(n^2)$,
$\log |H| = O(\log (n +k))$. 
Hence the amortized time  for the $k$-th iteration is  $O(\min\{k, n\} +  \log (n+k) )$.

For the worst-case bound we will use a binary heap. Then, the initialization takes $O(m)$ time to build a heap on the $m$ edges,
and the $k$-th iteration  costs
$O(\min \{k,n\} \cdot \log (n+k))$ for the heap operations.
\hfill $\square$

\section{Hardness Results}\label{sec:hardness}

We start with the definition of an {\it $f(m,n)$ reduction}.

\begin{definition}\label{def:mn-red}
Given graph problems $P$ and $Q$, an \emph{$f(m,n)$ reduction}, $P \leq_{f(m,n)} Q$, means that an input 
$G=(V,E)$ to $P$ with $|V|=n$,
$|E|=m$ can be reduced in $O(f(m,n))$ time to an input $G'=(V',E')$ to $Q$ 
such that from a solution for $Q$ on $G'$ we can obtain a solution for $P$ on $G$ in $O(f(m,n))$ time.
\end{definition}

The following lemma is straightforward.

\begin{lemma}
If $P \leq_{f(m,n)} Q$ then for any $f'(m,n) = \Omega (f(m,n))$, an $f'(m,n)$ algorithm for $Q$ implies an $f'(m,n)$ algorithm for $P$.
\end{lemma}

We mainly consider $f(m,n) =O(m+n)$, 
except for one reduction with $f(m,n) = (m+n) \cdot \log n$.

We now give some  $(m+n)$ reductions from Min-Wt-Cyc to several versions of the SiSP and SiSC problems.
Recall that Min-Wt-Cyc is the problem of finding a minimum weight cycle in a directed graph with
non-negative edge weights.

\begin{lemma}	\label{lemma:kSiSC-SiSP}
$k$-SiSC $\equiv_{(m+n)}$ $k$-SiSP.
\end{lemma}

\begin{proof}
The reduction from $k$-SiSC to $k$-SiSP is the same as that used in the algorithm for
$k$-SiSC in Section~\ref{sec:k-cycles}; we include it again here for completeness.
Suppose we are given an instance of the $k$-SiSC problem, a directed graph $G = (V,E)$ where for some $x \in V$, we want to find $k$-SiSCs passing through vertex $x$.
We can reduce this problem to $k$-SiSP by forming the graph $G'$ where,  we replace vertex $x$ in $G$ by vertices $x_i$ and $x_o$ in $G'$,  and we replace each edge $(u,x)$ in $G$ by an edge $(u_o,x_i)$ in $G'$ (and hence we also replace
each edge $(x,v)$ in $G$ by an edge $(x_o, v_i)$ in $G'$). It is not difficult to see that the $k$-th simple shortest path from $x_o$ to $x_i$ in $G'$ corresponds to the $k$-th simple
shortest cycle through $x$ in $G$. 

As the number of vertices and edges in $G'$ are linear in the number of vertices and edges, respectively, in $G$, we deduce that $k$-SiSC $\leq_{(m+n)}$ $k$-SiSP.

Now suppose that we are given an instance of the $k$-SiSP problem, a directed graph $G = (V,E)$ where for some $x,y \in V$, we want to find $k$-SiSPs from $x$ to $y$.
We can reduce this problem to $k$-SiSC by forming the graph $G'$ where, we add a new vertex $z$ and we place a directed edge of weight 0 from $y$ to $z$ and from $z$ to $x$. Now we can 
readily
see that the $k$-th simple shortest cycle through $z$ in $G'$ corresponds to the $k$-th simple shortest path from $x$ to $y$ in $G$. Hence, we obtain the desired result.
\end{proof}

It is shown in \cite{WW10} 
that 2-SiSP is  at least as hard as APSP for sub-cubic computations, using 
a reduction from minimum weight triangle.  That reduction is an $(m+n)$ reduction. However, a minimum weight triangle in a sparse graph can be
found in $O(m^{3/2})$ time using the triangle finding algorithm in~\cite{IR78}. Here we give an $(m+n)$ reduction from Min-Wt-Cyc to 2-SiSP
to establish that  a `sub-$mn$' algorithm for 2-SiSP would imply a similar improvement for Min-Wt-Cyc,
a long-standing open question.

\begin{lemma}	\label{lemma:MWC-2SiSP}
Min-Wt-Cyc $\leq_{(m+n)}$ 2-SiSP
\end{lemma}

\begin{proof}
Suppose we are given an instance of the Min-Wt-Cyc problem, a directed graph $G = (V,E)$ with vertex set $V  = \{1, 2, \ldots, n\}$, 
and we need
  to find the minimum weight cycle in the graph. 
 We will 
reduce this instance of the problem to that of computing 2-SiSP in a weighted directed graph,
as follows.

We
first construct 
the
 directed graph $G' = (V', E')$, as described in the proof of Lemma \ref{lemma:kSiSC-SiSP}.

Then 
we
create a directed graph $G'' = (V'',E'')$ such that it contains $G'$ as 
a
subgraph and 
also contains
a path $P$ ($p_0 \rightarrow p_1 \rightsquigarrow p_{n-1} \rightarrow p_n$) of $n+1$ vertices such that all edges on $P$ have weight 0.

Let 
$W=n \cdot w$,
where $w$ is the maximum weight of any edge in $G$.  
For
each $1 \leq j \leq n$, 
we
add 
an
edge of weight $(n-j+1)W$ from $p_{j-1}$ to $j_o$ and an edge of weight $jW$ from $j_i$ to $p_j$.  

Figure \ref{fig:hardF} depicts the full construction of $G''$ for $n=3$.

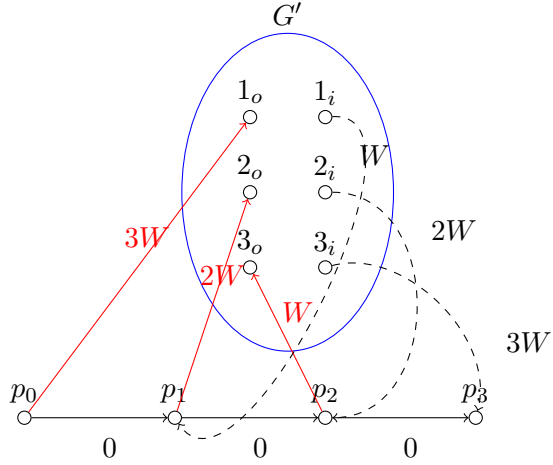
\begin{figure}
    \centering
        \makebox[.3\textwidth]{
            \begin{tikzpicture}[every node/.style={circle, draw, inner sep=0pt, minimum width=5pt}]
            \node (1o)[label=above:$1_o$] at (0,0)  {};
            \node (1i)[label=above:$1_i$] at (1,0) {};
            \node (2o)[label=above:$2_o$]	at (0,-1)		{};
            \node (2i)[label=above:$2_i$] at (1,-1)	{};
            \node (3o)[label=above:$3_o$] at (0,-2)	{};
            \node (3i)[label=above:$3_i$] at (1,-2) {};
            \node (p0)[label=above:$p_0$] at (-3,-4) {};
            \node (p1)[label=above:$p_1$] at (-1,-4) {};
            \node (p2)[label=above:$p_2$] at (1,-4) {};
            \node (p3)[label=above:$p_3$] at (3,-4) {};
            \node[ellipse, color=blue, fit=(1o)(1i)(2o)(2i)(3o)(3i), inner sep=4mm](G') [label=above:$G'$] {};
            \begin{scope}[->,every node/.style={ right}]
            \draw[->] (p0) -- (p1) node[midway, label=below:$0$] {};
            \draw[->] (p1) -- (p2) node[midway, label=below:$0$] {};
            \draw[->] (p2) -- (p3) node[midway, label=below:$0$] {};
            \draw[->,color=red] (p0) -- (1o) node[midway, label=above:$3W$] {};
            \draw[->,color=red] (p1) -- (2o) node[midway, label=above:$2W$] {};
            \draw[->,color=red] (p2) -- (3o) node[midway, label=above:$W$] {};
            \draw[->,color=black] (1i) to  [bend left=120, distance=15mm]  (p1) [dashed] node[label=left:$W$] at (2,-0.5) {};
            \draw[->,color=black] (2i) to  [bend left=90, distance=15mm]  (p2) [dashed] node[label=right:$2W$] at (2,-1.5) {};
            \draw[->,color=black] (3i) to  [bend left=60, distance=10mm]  (p3) [dashed] node[label=right:$3W$] at (3,-3) {};
            \end{scope}
            \end{tikzpicture}}
    \caption{Construction of $G''$ for $n = 3$ for Lemma~\ref{lemma:MWC-2SiSP}.}
    \label{fig:hardF}
\end{figure}

Now the 2-SiSP from $p_0$ to $p_{n}$ is of the form: $p_0 \rightsquigarrow p_{s-1} \rightarrow s_o \rightsquigarrow t_i \rightarrow p_t \rightsquigarrow p_n$
since it must contain a single detour. Further, $t > s-1$ since
 the path is simple. 
We claim that $t = s$. If not, then $t > s$ and
the weight of the path is at least $(n+2)W$. 
However,
any path of the form $p_0 \rightsquigarrow p_{s-1} \rightarrow s_o \rightsquigarrow s_i \rightarrow p_s \rightsquigarrow p_n$
has weight
strictly less than $(n+2)W$,
since any simple path in $G'$ has weight less than $W$, 
Hence, $t = s$ as 
long as
there is at least one 
path of the form $x_o \rightsquigarrow x_i$ (where $x \in V$) in $G'$.

Thus the 2-SiSP in $G''$ corresponds to the shortest path in $G'$ of the form $x_o \rightsquigarrow x_i$, which in turn corresponds to the minimum weight cycle in the original graph $G$. 

As the number of vertices and edges in $G''$ is
linear in the number of vertices and edges, respectively, in $G$, we 
obtain
the desired result.
\end{proof}

\begin{lemma} \label{lemma:kAVSiSC}
Min-Wt-Cyc $\leq_{(m+n)}$ $k$-AVSiSC.
\end{lemma}

\begin{proof}
Suppose  we are given an instance of the Min-Wt-Cyc problem, a directed graph $G=(V,E)$, and we need to find the minimum weight 
cycle in the graph.

We can find the $k$ minimum weight cycles passing through each vertex $x \in V$ by computing $k$-AVSiSC on $G$. We can then find the minimum weight cycle by taking the minimum of the shortest simple cycles passing through every vertex $x \in V$. Thus, we obtain the desired result.
\end{proof}

We establish two more hardness results for  the following two  problems:

\begin{itemize}
\item[(a)] {\it  $k$-th-All-SiSC} is the problem of computing the $k$-th simple shortest cycle in $G$ after the $k-1$ simple
shortest cycles in $G$ have been computed (for any constant $k>1$).

\item[(b)]{\it Second-APSiSP} is the problem of generating the second simple shortest path for all pairs of
vertices after APSP has been computed.

\end{itemize}

\begin{lemma}	\label{lemma:MWC-kASiSC}
Min-Wt-Cyc $\leq_{(m+n)}$ $k$-th-All-SiSC.
\end{lemma}

\begin{proof}
Suppose we are given an instance of the Min-Wt-Cyc, a directed graph $G = (V,E)$. Now we'll reduce this instance of the problem to that of computing $k$-th-All-SiSC in a weighted directed graph.

Now create a directed graph $G' = (V',E')$ such that it contains $G$ as its subgraph and $2(k-1)$ additional vertices coming from the vertex partitions $D = \{d_i\}_{i=1}^{k-1}$ and $E = \{e_i\}_{i=1}^{k-1}$. Fix some $x \in V$. For each $1\leq i\leq k-1$, add edges of weight 0 from $x$ to $d_i$, from $d_i$ to $e_i$ and from $e_i$ to $x$. 

Figure \ref{fig:hardB} depicts the full construction of $G'$ for $k = 3$.

\begin{figure}
    \centering
        \makebox[.3\textwidth]{
            \begin{tikzpicture}[every node/.style={circle, draw, inner sep=0pt, minimum width=5pt}]
            \node (x)[label=above:$x$] at (0,0)  {};
            \node[ellipse, fit=(x), inner sep=6mm](G) [label=above:$G$] {};
            \node (e1)[label=above:$e_1$] at (1,-2) {};
            \node (e2)[label=above:$e_2$] at (1,-3) {};
            \node (d1)[label=above:$d_1$] at (-1,-2) {};
            \node (d2)[label=above:$d_2$] at (-1,-3) {};
            \draw[->] (x) -- (d1);
            \draw[->] (d1) -- (e1);
            \draw[->] (e1) -- (x);
            \draw[->] (x) -- (d2);
            \draw[->] (d2) -- (e2);
            \draw[->] (e2) -- (x);
            \end{tikzpicture}}
    \caption{Construction of $G'$ for $k = 3$ for Lemma~\ref{lemma:MWC-kASiSC}.}
    \label{fig:hardB}
\end{figure}
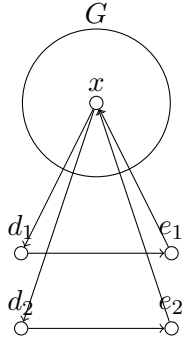

Now the first $(k-1)$ min-weight cycles in $G'$ correspond to the cycles involving vertices $x$, $d_i$ and $e_i$ (for each $1\leq i\leq k-1$). And the $k$-th min-weight cycle in $G'$ corresponds to the minimum weight cycle in $G$.

As the number of vertices and edges in $G'$ are linear in the number of vertices and edges, respectively, in $G$,
 we get the desired result.
\end{proof}

\begin{lemma}	\label{thm:UW-2APSiSP}
APSP $\leq_{(m+n)}$ Second-APSiSP.
\end{lemma}

\begin{proof}
Suppose we are given an arbitrary directed graph $G = (V_G, E_G)$ where $V_G = \{1, 2, \ldots, n \}$. Now we'll reduce the problem of computing APSP on $G$ to one of computing Second-APSiSP in another weighted directed graph.

Now construct a graph $G' = (V', E')$ on $3n + 1$ nodes that contains $G$ as its subgraph. Apart from the $n$ vertices present in 
$G$, $G'$ also contains a vertex $s$ and $2n$ additional vertices coming from the vertex partitions $A = \{a_i\}_{i=1}^{n}$ and $B = \{b_i\}_{i=1}^{n}$.

For each $1\leq i\leq n$, add an edge of weight 0 from $a_i$ to $s$ and from $s$ to $b_i$.

For every $1 \leq i\leq n$, also add edges of weight 1 from $a_i$ to $i$ and from $i$ to $b_i$.

Figure \ref{fig:hardD} depicts the full construction of $G'$ for $n = 3$.

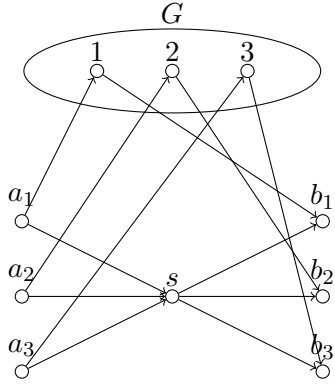
\begin{figure}
    \centering
        \makebox[.3\textwidth]{
            \begin{tikzpicture}[every node/.style={circle, draw, inner sep=0pt, minimum width=5pt}]
            \node (a1)[label=above:$a_1$] at (0,0)  {};
            \node (a2)[label=above:$a_2$] at (0,-1) {};
            \node (a3)[label=above:$a_3$]	at (0,-2) {};
            \node (s)[label=above:$s$] at (2,-1) {};
            \node (b1)[label=above:$b_1$] at (4,0)	{};
            \node (b2)[label=above:$b_2$] at (4,-1)	{};
            \node (b3)[label=above:$b_3$] at (4,-2) {};
            \node (1)[label=above:$1$] at (1,2) {};
            \node (2)[label=above:$2$] at (2,2)	{};
            \node (3)[label=above:$3$] at (3,2)	{};
            \node[ellipse, fit=(1)(2)(3), inner sep=3mm](G) [label=above:$G$] {};
            \draw[->] (a1) -- (1);
            \draw[->] (a1) -- (s);
            \draw[->] (a2) -- (2);
            \draw[->] (a2) -- (s);
            \draw[->] (a3) -- (3);
            \draw[->] (a3) -- (s);
			 \draw[->] (1) -- (b1);
			 \draw[->] (2) -- (b2);
			 \draw[->] (3) -- (b3);
			 \draw[->] (s) -- (b1);
			 \draw[->] (s) -- (b2);
			 \draw[->] (s) -- (b3);
            \end{tikzpicture}}
    \caption{Construction of $G'$ for $n = 3$ for Lemma~\ref{thm:UW-2APSiSP}.}
    \label{fig:hardD}
\end{figure}

Now the 2-SiSP from some $a_i$ to some $b_j$ (where $1\leq i,j \leq n$) is of the form :- $a_i \rightarrow i \rightsquigarrow j \rightarrow b_j$, where the second simple shortest path first takes the edge $(a_i,i)$ and then it takes the shortest path from $i$ to $j$ and then the edge $(j,b_j)$.

Thus every 2-SiSP in $G'$ from a vertex $a_i$ in $A$ to a vertex $b_j$ in $B$ corresponds to a shortest path from $i$ to $j$ in the original graph $G$.

As the number of vertices and edges in $G'$ are linear in the number of vertices and edges, respectively, in $G$, we get the desired result.
\end{proof}

The above reductions show that for any $k\geq 2$, $k$-SiSP, $k$-SiSC, $k$-AVSiSC and $k$-th-All-SiSC cannot be solved in $o(m \cdot n)$ time 
unless an improved algorithm is  obtained for Min-Wt-Cyc.
Also, the last reduction shows that computing Second-APSiSP is at least as hard as computing APSP. (It 
can also be seen
that  Min-Wt-Cyc $\leq_{(m+n)}$ APSP.)

\vone
\noindent
{\bf Unweighted Graphs.} Most of our reductions go through (either unchanged or with small changes) for unweighted graphs. The one exception is Min-Wt-Cyc to 2-SiSP. Here in fact, there is a
randomized 
$\tilde{O}(k \cdot m \sqrt n)$
 time algorithm for $k$-SiSP~\cite{RZ12}.  For the reductions from cycle problems to
path problems 
($k$-SiSC to $k$-SiSP and $2$-AVSiSC to $2$-APSiSP),
 we used an edge of weight 0 from $v_i$ to $v_o$. In the unweighted case, we can leave the edge weight at 1, and observe that this preserves the ordering of simple shortest paths for any
pair of vertices, since a path of length $r$ in $G$ from $s$ to $t$ is now transformed into a path
of length $2r-1$ from $s_o$ to $t_i$ in $G'$.

\subsection{Undirected Graphs}

Our algorithms for $k$-APSiSP and $k$-All-SiSP work for undirected graphs. 
Hence $k$-All-SiSP and 2-APSiSP have the same time bound as for the directed case.
However, $k$-SiSP in undirected graphs can be solved in 
$\tilde{O}(m)$
time \cite{KIM82}, hence for $k \geq 3$, $k$-APSiSP can be computed
in $\tilde{O}(mn^2)$ time in undirected graphs.

Our reduction  from $k$-SiSC to $k$-SiSP problem given in Section~\ref{sec:k-cycles} does not work for undirected graphs.
We give an alternate reduction (with a small $O(\log n)$ increase in the bound). 
\begin{lemma}	\label{lemma:unkSiSCkSiSP}
$k$-SiSC $\leq_{((m+n) \cdot \log n)}$ $k$-SiSP.
\end{lemma}
\begin{proof}
Let the input be $G = (V,E)$ and the vertex $x \in V$, for which we want to compute $k$-SiSCs. We assume that the vertices are labeled from $1$ to $n$. 
We first show that $k$-SiSC in $G$ can be computed with $\lceil \log n \rceil$ calls to $k$-SiSP.
Let $\mathcal{N}(x)$ be the neighbor-set of $x$. We create $\lceil \log n \rceil$ graphs $G_i = (V_i, E_i)$ such that $\forall 1 \leq i\leq \lceil \log n\rceil$, $G_i$ contains two additional vertices $x_{0,i}$ and $x_{1,i}$ (instead of the vertex $x$) and $\forall y \in \mathcal{N}(x)$, the edge $(y,x_{0,i}) \in E_i$ if $y's$ $i$-th bit is $0$, otherwise the edge $(y,x_{1,i}) \in E_i$. This takes $O((m+n) \cdot \log n)$ time and we observe that every cycle through $x$ will appear as a path from $x_{0,i}$ to $x_{1,i}$ in at least one of the $G_i$. Hence, the $k$-th shortest path in the collection of $k$-SiSPs from $x_{0,i}$ to  $x_{1,i}$ in $G_i$ $\forall  1 \leq i \leq \lceil \log n \rceil$ (after removing duplicates), corresponds to the $k$-th SiSC passing through $x$. 
If we create new vertices $z$ and $z'$, connect $z$ to the $x_{0,i}$ vertices and $z'$ to the
$x_{1,i}$, then computing $k'$-SiSP between $z$ and $z'$ in this graph for $k'=k \cdot \lceil \log n\rceil$, gives us $k$-SiSC 
through $x$ in $G$ as shown above.
\end{proof}

Using the above lemma
and the results in Sections~\ref{sec:k-cycles} and  \ref{sec:allsisp1} we can compute 
$k$-SiSC in $\tilde{O}(km)$ time, 
$k$-AVSiSC in $\tilde{O}(kmn)$ time and
$k$-All-SiSC in $\tilde{O}(m)$ time per cycle after a startup cost of $\tilde{O}(mn)$ in undirected graphs.

Most of our hardness results (from $k$-SiSP to $k$-SiSC, Min-Wt-Cyc to $k$-AVSiSC, Min-Wt-Cyc to $k$-All-SiSC) also hold for undirected graphs. However our reduction from Min-Wt-Cyc to $2$-SiSP for directed graphs does not hold for undirected graphs. This is not surprising as 2-SiSP can be computed in 
$\tilde{O}(m)$ time \cite{KIM82}.

\subsection{Discussion}  

There are several important problems
on sparse graphs for which $\tilde{O}(mn)$ is the current best time bound:
 Min-Wt-Cyc, APSP (for both problems, either directed or undirected, and either weighted or 
 unweighted), weighted $k$-SiSP, 
 and the collection of weighted directed graph problems for which we have given $\tilde{O}(mn)$ time 
algorithms in this paper.
This suggests that the class of problems that
currently have $\tilde{O}(mn)$ time algorithms is an important one, with Min-Wt-Cyc being the key problem,
similar to APSP for cubic computations, and 3SUM for quadratic computations.

\section{Conclusion}

We have presented new algorithms to compute $k$  simple shortest paths and cycles in a weighted directed 
 (or undirected)
graph, complementing many of our upper bounds with hardness results for sparse graphs
(by reductions from Min-Wt-Cyc).
 Our results include the following.

\vmone
\begin{itemize}
\item  A 2-APSiSP algorithm  which almost matches the current best 
$O(mn + n^2\log n\log n)$ bound for finding the two simple shortest paths for just a single pair of vertices.
 
 \item A new recursive algorithm to compute $k$-APSiSP, which improves the best prior bound for $k=3$
 (for directed graphs)
 ;
 although this algorithm, {\sc APSiSP}, does not give improved bounds for $k>3$
 (and for $k>2$ for undirected graphs)
 , it presents a new method for
 finding $k$ shortest paths, and leaves open the possibility for further improvement,  if a better
 algorithm can be found to compute the nearly $k$ SiSP sets $Q_k(x,y)$.

 \item Algorithms and hardness results for the simple cycles versions, $k$-SiSC and $k$-AVSiSC.
 
 \item  Algorithms to efficiently
enumerate simple paths and simple cycles in $G$ in nondecreasing order of weight, and a
conditional  hardness result that enumerating simple cycles in nondecreasing order of weights is a 
significantly harder problem than a similar enumeration of simple paths.
  \end{itemize}
  
  We conclude with some avenues for further research.
  
 1. The main open question for $k$-APSiSP is to come up with faster algorithms to compute the $Q_k (x,y)$ sets for larger values of $k$. This is the key to a
 faster $k$-APSiSP algorithm using our approach, for $k >2$.
  
  2. The space requirements of our all-pairs algorithms are high. Can we come up with
  space-efficient algorithms that match our time bounds?
  
  3. Can 
  we come up with other  hardness results for sparse graphs, for example, 
 can we show that Min-Wt-Cyc $\leq_{(m+n)}$ APSP in {\it undirected graphs}?
  (For directed graphs there is a simple reduction.)

\bibliographystyle{abbrv}
\bibliography{references}

\begin{thebibliography}{10}

\bibitem{Bernstein10}
A.~Bernstein.
\newblock A nearly optimal algorithm for approximating replacement paths and k
  shortest simple paths in general graphs.
\newblock In {\em Proceedings of the twenty-first annual ACM-SIAM symposium on
  Discrete Algorithms}, pages 742--755. Society for Industrial and Applied
  Mathematics, 2010.

\bibitem{BK09}
A.~Bernstein and D.~Karger.
\newblock A nearly optimal oracle for avoiding failed vertices and edges.
\newblock In {\em Proceedings of the forty-first annual ACM symposium on Theory
  of computing}, pages 101--110. ACM, 2009.

\bibitem{CLRS09}
T.~H. Cormen, C.~E. Leiserson, R.~L. Rivest, and C.~Stein.
\newblock {\em Introduction to Algorithms, Third Edition}.
\newblock The MIT Press, 3rd edition, 2009.

\bibitem{DI04}
C.~Demetrescu and G.~F. Italiano.
\newblock A new approach to dynamic all pairs shortest paths.
\newblock {\em J. ACM}, 51(6):968--992, 2004.

\bibitem{DI06}
C.~Demetrescu and G.~F. Italiano.
\newblock Experimental analysis of dynamic all pairs shortest path algorithms.
\newblock {\em ACM Transactions on Algorithms (TALG)}, 2(4):578--601, 2006.

\bibitem{DTCR08}
C.~Demetrescu, M.~Thorup, R.~A. Chowdhury, and V.~Ramachandran.
\newblock Oracles for distances avoiding a failed node or link.
\newblock {\em SIAM Journal on Computing}, 37(5):1299--1318, 2008.

\bibitem{Eppstein98}
D.~Eppstein.
\newblock Finding the k shortest paths.
\newblock {\em SIAM Journal on Computing}, 28(2):652--673, 1998.

\bibitem{Feng14}
G.~Feng.
\newblock Finding k shortest simple paths in directed graphs: A node
  classification algorithm.
\newblock {\em Networks}, 64(1):6--17, 2014.

\bibitem{FR15}
A.~Frieder and L.~Roditty.
\newblock An experimental study on approximating k shortest simple paths.
\newblock {\em Journal of Experimental Algorithmics (JEA)}, 19:1--5, 2015.

\bibitem{GL09}
Z.~Gotthilf and M.~Lewenstein.
\newblock Improved algorithms for the k simple shortest paths and the
  replacement paths problems.
\newblock {\em Information Processing Letters}, 109(7):352--355, 2009.

\bibitem{HC99}
E.~Hadjiconstantinou and N.~Christofides.
\newblock An efficient implementation of an algorithm for finding k shortest
  simple paths.
\newblock {\em Networks}, 34(2):88--101, 1999.

\bibitem{HMS07}
J.~Hershberger, M.~Maxel, and S.~Suri.
\newblock Finding the k shortest simple paths: A new algorithm and its
  implementation.
\newblock {\em ACM Transactions on Algorithms (TALG)}, 3(4):45, 2007.

\bibitem{HS01}
J.~Hershberger and S.~Suri.
\newblock Vickrey prices and shortest paths: What is an edge worth?
\newblock In {\em Foundations of Computer Science, 2001. Proceedings. 42nd IEEE
  Symposium on}, pages 252--259. IEEE, 2001.

\bibitem{HSB07}
J.~Hershberger, S.~Suri, and A.~Bhosle.
\newblock On the difficulty of some shortest path problems.
\newblock {\em ACM Transactions on Algorithms (TALG)}, 3(1):5, 2007.

\bibitem{IR78}
A.~Itai and M.~Rodeh.
\newblock Finding a minimum circuit in an graph.
\newblock {\em SIAM Journal on Computing}, 7(4):413--423, 1978.

\bibitem{Joh75}
D.~B. Johnson.
\newblock Finding all the elementary circuits of a directed graph.
\newblock {\em SIAM Journal on Computing}, 4(1):77--84, 1975.

\bibitem{Joh77}
D.~B. Johnson.
\newblock Efficient algorithms for shortest paths in sparse networks.
\newblock {\em Journal of the ACM}, 24(1):1--13, 1977.

\bibitem{KKP93}
D.~R. Karger, D.~Koller, and S.~J. Phillips.
\newblock Finding the hidden path: Time bounds for all-pairs shortest paths.
\newblock {\em SIAM J. Comput.}, 22(6):1199--1217, 1993.

\bibitem{KIM82}
N.~Katoh, T.~Ibaraki, and H.~Mine.
\newblock An efficient algorithm for k shortest simple paths.
\newblock {\em Networks}, 12(4):411--427, 1982.

\bibitem{LP97}
K.~N. Lalgudi and M.~C. Papaefthymiou.
\newblock Computing strictly-second shortest paths.
\newblock {\em Information Processing Letters}, 63(4):177--181, 1997.

\bibitem{Lawler72}
E.~L. Lawler.
\newblock A procedure for computing the k best solutions to discrete
  optimization problems and its application to the shortest path problem.
\newblock {\em Management Science}, 18(7):401--405, 1972.

\bibitem{Lawler77}
E.~L. Lawler.
\newblock Comment on a computing the k shortest paths in a graph.
\newblock {\em Communications of the ACM}, 20(8):603--605, 1977.

\bibitem{MP03}
E.~Q. Martins and M.~M. Pascoal.
\newblock A new implementation of {Y}en's ranking loopless paths algorithm.
\newblock {\em Quarterly Journal of the Belgian, French and Italian Operations
  Research Societies}, 1(2):121--133, 2003.

\bibitem{MK05}
W.~Matthew~Carlyle and R.~Kevin~Wood.
\newblock Near-shortest and k-shortest simple paths.
\newblock {\em Networks}, 46(2):98--109, 2005.

\bibitem{Minieka74}
E.~Minieka.
\newblock On computing sets of shortest paths in a graph.
\newblock {\em Communications of the ACM}, 17(6):351--353, 1974.

\bibitem{Perko86}
A.~Perko.
\newblock Implementation of algorithms for k shortest loopless paths.
\newblock {\em Networks}, 16(2):149--160, 1986.

\bibitem{Pettie04}
S.~Pettie.
\newblock A new approach to all-pairs shortest paths on real-weighted graphs.
\newblock {\em Theoretical Computer Science}, 312(1):47 -- 74, 2004.

\bibitem{Roditty10}
L.~Roditty.
\newblock On the k shortest simple paths problem in weighted directed graphs.
\newblock {\em SIAM Journal on Computing}, 39(6):2363--2376, 2010.

\bibitem{RZ12}
L.~Roditty and U.~Zwick.
\newblock Replacement paths and k simple shortest paths in unweighted directed
  graphs.
\newblock {\em ACM Transactions on Algorithms (TALG)}, 8(4):33, 2012.

\bibitem{Sedeno12}
A.~Sede{\~n}o-Noda.
\newblock An efficient time and space k point-to-point shortest simple paths
  algorithm.
\newblock {\em Applied Mathematics and Computation}, 218(20):10244--10257,
  2012.

\bibitem{Tar73}
R.~Tarjan.
\newblock Enumeration of the elementary circuits of a directed graph.
\newblock {\em SIAM Journal on Computing}, 2(3):211--216, 1973.

\bibitem{Tie70}
J.~C. Tiernan.
\newblock An efficient search algorithm to find the elementary circuits of a
  graph.
\newblock {\em Comm. of the ACM (CACM)}, 13:722--726, 1970.

\bibitem{Wei72}
H.~Weinblatt.
\newblock A new search algorithm to find the elementary circuits of a graph.
\newblock {\em Journal of the ACM (JACM)}, 19:43--56, 1972.

\bibitem{WW10}
V.~V. Williams and R.~Williams.
\newblock Subcubic equivalences between path, matrix and triangle problems.
\newblock In {\em Foundations of Computer Science (FOCS), 2010 51st Annual IEEE
  Symposium on}, pages 645--654. IEEE, 2010.

\bibitem{Yen71}
J.~Y. Yen.
\newblock Finding the k shortest loopless paths in a network.
\newblock {\em Management Science}, 17(11):712--716, 1971.

\end{thebibliography}

\end{document}